\documentclass[reprint,nofootinbib,superscriptaddress,amsmath,amssymb,aps,]{revtex4-2}

\usepackage{graphicx} 
\usepackage{dcolumn}
\usepackage{xr}
\usepackage{bm}
\usepackage{lipsum}
\usepackage{wrapfig,lipsum,booktabs}
\usepackage[utf8]{inputenc} 
\usepackage[T1]{fontenc}    
\usepackage{hyperref}       
\usepackage{url}            
\usepackage{booktabs}       
\usepackage{amsmath}
\usepackage{amsfonts} 
\usepackage{amssymb}
\usepackage{nicefrac}       
\usepackage{microtype}      
\usepackage{xcolor}         
\usepackage{bm}
\usepackage{wrapfig}
\usepackage{comment}
\usepackage{mathtools}
\usepackage{amsthm}
\usepackage[capitalise]{cleveref}
\crefname{appendix}{App.}{Apps.}
\usepackage{autonum}

\labelformat{subsection}{\thesection#1}
\newtheorem{theorem}{Theorem}[subsection]
\newtheorem{corollary}[theorem]{Corollary}
\newtheorem{proposition}[theorem]{Proposition}

\usepackage{algorithm}
\usepackage[noend]{algpseudocode}

\begin{document}

\title{Learning biophysical models of gene regulation with probability flow matching}
\author{Suryanarayana Maddu}
\altaffiliation{These authors contributed equally to this work}
\affiliation{Center for Computational Biology, Flatiron Institute, New York, NY, USA, 10010}
\affiliation{Molecular $\&$ Cellular
Biology Department, Harvard University, Cambridge, MA 02138}
\author{Victor Chard\`es}
\altaffiliation{These authors contributed equally to this work}
\affiliation{Center for Computational Biology, Flatiron Institute, New York, NY, USA, 10010}
\affiliation{Molecular $\&$ Cellular
Biology Department, Harvard University, Cambridge, MA 02138}
\author{Michael J. Shelley}
\affiliation{Center for Computational Biology, Flatiron Institute, New York, NY, USA, 10010}
\affiliation{Courant Institute of Mathematical Sciences, New York University, New York, NY, USA, 10012}

\begin{abstract}

Cellular differentiation is governed by gene regulatory networks, the high-dimensional stochastic biochemical systems that determine the transcriptional landscape and mediate cellular responses to signals and perturbations. Although single-cell RNA sequencing provides quantitative snapshots of the transcriptome, current methods for inferring gene-regulatory dynamics often lack mechanistic interpretability and fail to generalize to unseen conditions. Here we introduce Probability Flow Matching (PFM), a scalable framework for learning biophysically consistent stochastic processes directly from time-resolved single-cell measurements. Applying PFM to three hematopoiesis datasets, we show that models with similar interpolation accuracy can encode fundamentally different dynamics, with only biophysically consistent formulations accurately capturing mechanisms of lineage transitions, fate specification, and gene perturbation responses. We further demonstrate that PFM accommodates unbalanced populations, enabling simultaneous inference of cellular proliferation and death dynamics. Together, these results establish PFM as a flexible, scalable framework for integrating mechanistic modeling with single-cell omics.

\end{abstract}
\maketitle

The molecular programs operating in single cells during homeostasis and the changes these programs undergo during differentiation, development, and disease progression are central research themes of contemporary biology \cite{moris2016transition, trapnell2015defining}. Nonetheless, we remain far from a comprehensive understanding of how molecular programs are implemented during cellular differentiation and cell fate decisions, even in well-characterized systems such as hematopoiesis \cite{velten2017human, tusi2018population}. Recent advances in single-cell multi-omics technologies have revolutionized our ability to profile the molecular state of individual cells, capturing diverse features such as gene expression, chromatin accessibility, and other regulatory layers, with unprecedented resolution \cite{macosko2015highly, buenrostro2015single, bintu2018super}. Notably, single-cell RNA-seq has already reshaped our understanding of cell fate decisions, challenging classical and longstanding models of cellular differentiation \cite{laurenti2018haematopoietic, sagar2020deciphering}. As these technologies increasingly enable the simultaneous measurements through multiple modalities within the same cell, there is a growing need to develop computational modeling frameworks to integrate and analyze these high-dimensional datasets \cite{gayoso2021joint, kang2022roadmap, cai2022machine}.\\

\begin{figure*}
  \centering
  \includegraphics[width = 17.5cm]{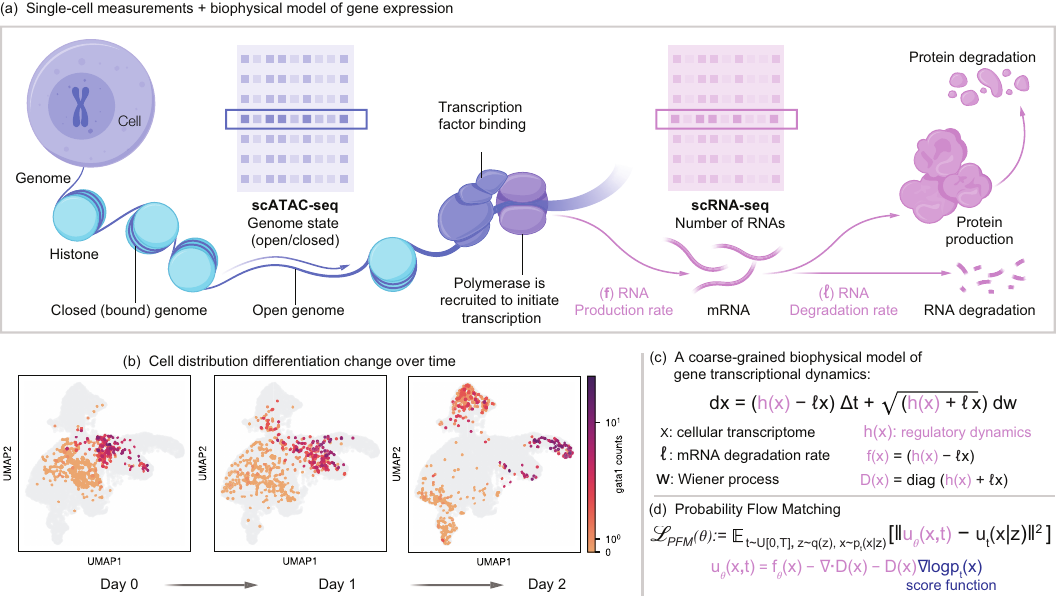}
\caption{
\textbf{Single-cell measurements and probability flow modeling of gene transcriptional dynamics.}
(a) Schematic showing chromatin accessibility (measured by scATAC-seq) to transcription initiation through transcription complex binding, RNA production (measured by scRNA-seq) and degradation, and downstream protein dynamics through synthesis and degradation.
(b) Example of how cell-type distributions evolve in the 2D UMAP coordinates during \textit{in vivo} differentiation of Hematopoietic Stem Cells (HSC) into terminal cell types \cite{weinreb2020lineage}.
(c,d) Probability Flow Matching (PFM) framework: a coarse-grained stochastic model of gene transcription is fit to observed marginal distributions from single-cell omics using the Probability Flow Matching formulation.
}
  \label{fig:schematic}
\end{figure*}

Most computational frameworks for single-cell analysis learn low-dimensional latent representations from high-dimensional data \cite{lopez2018deep, becht2019dimensionality, weinreb2020fundamental}, which enable downstream tasks such as cell type annotation, denoising, pseudotime {ordering}, and lineage inference \cite{trapnell2014dynamics, street2018slingshot, moon2019visualizing} and detection of rare or transitional states \cite{setty2016wishbone}. While these methods can effectively capture statistical structure and organize complex multi-modal datasets, they remain largely descriptive \cite{saelens2019comparison, lahnemann2020eleven}, and offer limited insight into the molecular mechanisms governing state maintenance and transitions \cite{procopio2023combined}. 
Mechanistic frameworks based on dynamical models such as systems of differential equations or probabilistic dynamical formulations \cite{la2018rna, weber2018identification, herbach2017inferring} can provide principled means to uncover these processes, but face significant challenges: single-cell data are high-dimensional, noisy, and typically cross-sectional \cite{saelens2019comparison, schiebinger2019optimaltransport}, making it difficult to reconstruct temporal dynamics like differentiation \cite{weinreb2020fundamental, qiu2017reversed, setty2016wishbone, maddu2025learning}, especially when accounting for population-level shifts due to proliferation and death.




Nonetheless, a wide range of computational methods have been proposed to model time-resolved snapshot data, particularly in the context of single-cell measurements. Among these, methods based on optimal transport have emerged as a powerful framework for constructing coupling between cell state distributions between time points \cite{yang2018scalable, schiebinger2019optimaltransport, schiebinger2021reconstructing, zhang2021optimal, bunne2023learning}. Frameworks such as dynamical optimal transport (OT) build continuous-time flows that interpolate between multiple snapshots, offering a way to reconstruct cellular trajectories across time \cite{tong2020trajectorynet, chizat2022trajectory, lavenant2023mathematical, bunne2023schrodinger}. Similar approaches based on differential equations models have been successful in reconstructing gene regulatory dynamics from high-dimensional single-cell snapshot data \cite{ocone2015reconstructing, matsumoto2017scode, sanchez2018bayesian, hashimoto2016learning, yeo2021generative}. Despite their wide adoption and interpolative accuracy, most of these methods begin by projecting the high-dimensional single-cell data onto a low-dimensional latent space (e.g., via PCA), and then inferring couplings between distributions in this reduced space \cite{tong2020trajectorynet, yeo2021generative}. While such projections offer computational scalability and simplify the modeling of noise, they often lack mechanistic interpretability. As such, they do not lend themselves to performing gene-level \textit{in silico} perturbations \cite{maddu2025learning}, and lack the means to incorporate biophysical constraints. In particular, models constrained to a manifold may generalize poorly to out-of-distribution initial conditions, limiting their applicability in settings involving experimental perturbations or unseen conditions.\\


Simulation-free training strategies based on Flow Matching (FM) \cite{lipman2022flow, lipman2024flow} offer a scalable alternative that enables inference directly in the gene space. However, existing formulations typically assume simplified noise structure and are limited to modeling dynamics between two marginal distributions \cite{tong2024simulation}. Extensions to multiple marginals \cite{rohbeck2025modeling} exist, but ignore stochasticity altogether while still working in reduced spaces. In our prior work on Probability Flow Inference (PFI) \cite{chardes2023stochastic, maddu2025learning, zhang2025inferring}, we proposed a framework for inferring gene-regulatory dynamics directly in gene space while explicitly accounting for intrinsic stochasticity and growth. We demonstrated that modeling intrinsic noise is essential for accurately recovering gene-regulatory interactions. However, PFI has issues with computational scalability, arising from the necessity to solve and optimize ODEs and compute the Wasserstein distance between empirical marginals. As a result, PFI is currently restricted to small gene sets (on the order of a few tens of carefully selected genes) and moderate sample sizes (up to a few thousand cells). An alternative class of methods such as RNA velocity  \cite{zheng2023pumping, ancheta2024challenges, la2018rna, bergen2020generalizing}, overcomes this computational bottleneck but makes a simplifying assumption that gene dynamics are decoupled, ignoring known nonlinear and coordinated regulatory interactions \cite{pratapa2020benchmarking}. This highlights the apparent trade-offs between computational scalability and biophysical consistency, which in turn lead to persistent challenges in interpretability, generalizability, and robustness \cite{richter2026generative}.

We avoid these compromises through Probability Flow Matching (PFM), a scalable and simulation-free framework for inferring mechanistic models of gene expression dynamics from multi-marginal single-cell omics data. PFM builds on recent advances in score-based generative modeling \cite{song2020score, song2021score} and flow matching methods \cite{chen2018neural, lipman2022flow, song2021score}, enabling direct regression of gene regulatory dynamics without expensive high-dimensional ODE solves, all while accommodating stochasticity of arbitrary nature. To ensure smoothness and numerical stability, we parameterize conditional Gaussian paths using spectrally regularized Chebyshev interpolants. This approach captures global temporal deformation across marginals and yields improved accuracy compared to standard linear \cite{tong2024simulation} or cubic spline-based interpolants \cite{rohbeck2025modeling}. Additionally, we extend the PFM framework to incorporate cell proliferation and death, allowing for simultaneous inference of cell-state-specific growth dynamics.\\



We benchmark our framework against simulation-based methods such as Probability Flow Inference on canonical test cases, demonstrating improved accuracy along with an order-of-magnitude reduction in computational and memory requirements. We then apply PFM to the well-studied dynamics of hematopoietic differentiation, exploring models defined by different force and diffusion parameterizations on both simulated and experimental single-cell datasets. Models that incorporate intrinsic molecular stochasticity not only interpolate the observed data with higher fidelity but also recover biologically meaningful gene regulatory interactions, consistent with our previous findings \cite{maddu2025learning}. Importantly, these biophysically grounded models generalize across experimental settings \cite{georgolopoulos2021discrete} and accurately predict cellular trajectories from yet unseen initial conditions, reflecting their robustness and predictive power. Finally, we extend our framework to infer cell-state–specific growth dynamics in an \textit{in vitro} hematopoietic dataset \cite{weinreb2020lineage}, predicting significant population changes over time while capturing the mechanisms of cell fate specification.

\begin{figure*}
  \centering
  \includegraphics[width = 16.cm]{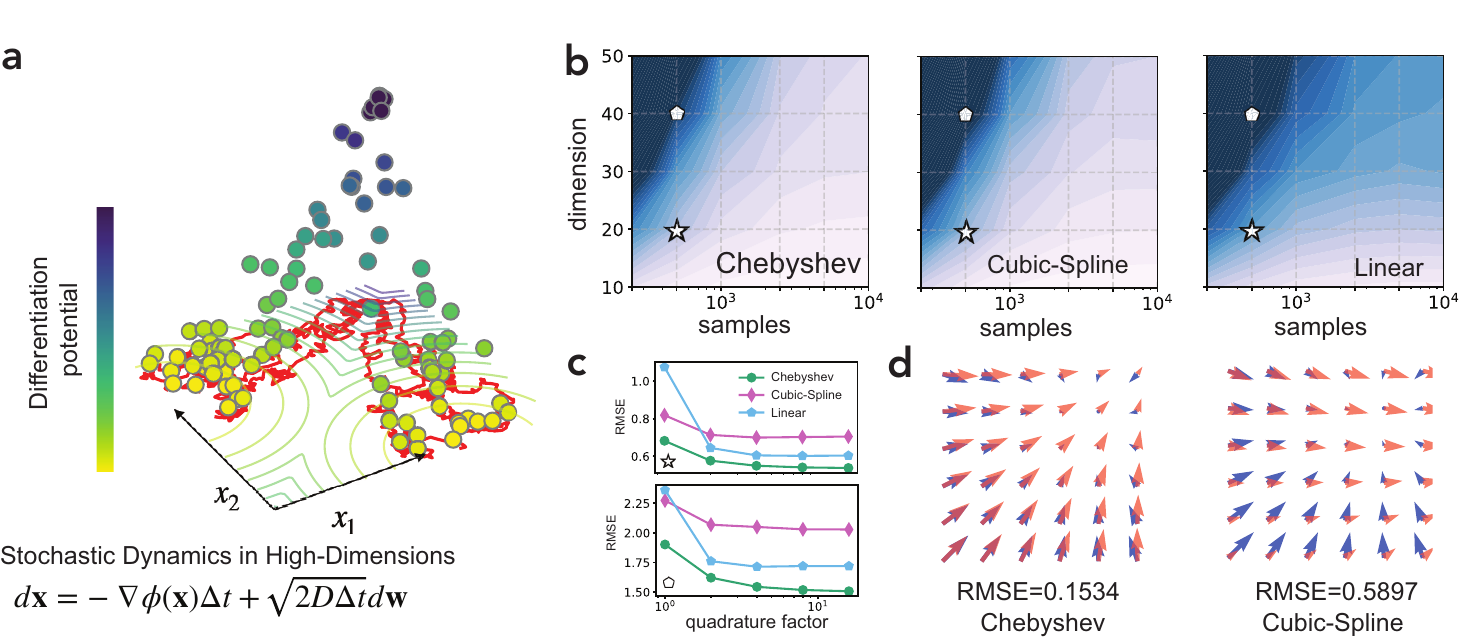}
\caption{
    \textbf{Evaluating the impact of choice of conditional mean paths $\bm{\mu}_t$.}
    \textbf{(a)} Illustration of a high-dimensional differentiation landscape with example stochastic trajectories (red) evolving according to the SDE $d\mathbf{x} = -\nabla \phi(\mathbf{x})\Delta t + \sqrt{2D\Delta t}d \mathbf{w}$. Points are colored by differentiation potential $\phi(\mathbf{x})$. (b) Accuracy of Chebyshev, cubic-spline, and linear interpolation as a function of state dimension and number of samples. Lighter regions indicate lower RMSE. The white stars mark the operating regime used in panel (c). Line plots (right) show RMSE as a function of quadrature resolution for each method.
(d) Reconstructed vector fields using Chebyshev (top) and cubic-spline (bottom) interpolation in a 2D slice of the state space. Chebyshev interpolation captures both magnitude and direction of the drift more accurately (RMSE = 0.1534) compared to cubic splines (RMSE = 0.5897).}
    \label{fig:Bifur-interpolation}
\end{figure*}

\section{Results}
\subsection{Probability Flow Matching: Overview and our formulation}


A broad class of inverse problems in computational biology involves recovering latent stochastic dynamics from cross-sectional snapshot data collected over few discrete time points. Examples include reconstructing transcriptional state dynamics from scRNA-seq data \cite{tong2020trajectorynet, maddu2025learning}, learning chromatin polymer models from FISH measurements \cite{shi2021hi}, and identifying PDE-based descriptions from spatial omics data \cite{raspopovic2014digit, qiu2024spatiotemporal}.\\



Cross-sectional snapshot data correspond to observations of the time-evolving density \(p(\mathbf{x}, t)\). In the context of gene regulatory networks, the state \(\mathbf{x}\) represents the transcriptional state of the cell, and the underlying stochastic dynamics are naturally described by the Fokker–Planck (FP) equation:
\begin{align}
\frac{\partial p_t(\mathbf{x})}{\partial t} &= -\nabla \cdot \left( \mathbf{u}_t(\mathbf{x}) p_t(\mathbf{x}) \right) + g_t(\mathbf{x})p_t(\mathbf{x}), \quad \textrm{where} \nonumber \\
    \mathbf{u}_t(\mathbf{x}) &= \mathbf{f(x)} - \nabla \! \cdot \! \mathbf{D(x)} - \mathbf{D(x)}\nabla \log p_t(\mathbf{x}). \label{eq:FP}
\end{align}
Here $\mathbf{f(x)}$ encodes gene-regulatory interactions, while $\mathbf{D(x)} = \mathbf{G(x)}\mathbf{G(x)}^\top \in \mathbb{S}_{+}^d$ is the diffusion matrix that captures intrinsic/extrinsic noise. Here, $g_t(\mathbf{x})$ models cell proliferation and death \cite{zhang2025inferring}. The velocity field $\mathbf{u}_t(\mathbf{x})$ is the right-hand side of the Probability Flow ODE \cite{song2020score}, which is found by recasting the FP equation into Lagrangian coordinates. Here, \( \nabla \log p_t(\mathbf{x}) \) denotes the score function \cite{hyvarinen2005estimation}. The FP formulation provides a unifying framework for transcriptomic dynamics, within which many existing models arise as special cases. Deterministic approaches such as TrajectoryNet \cite{tong2020trajectorynet} and TIGON \cite{sha2024reconstructing} correspond to \( \mathbf{D}(\mathbf{x}) = 0 \) with time-dependent drift \( \mathbf{f}(\mathbf{x}, t) \), while biophysically motivated models (e.g., PRESCIENT \cite{yeo2021generative}) impose structure such as conservative forces \( \mathbf{f}(\mathbf{x}) = -\nabla \phi(\mathbf{x}) \) with isotropic diffusion.\\





In practice, measurement data is usually given as $K$ statistically independent cross-sectional snapshots, each composed of {$n_k$} samples, taken from the true process at successive times $t_0  < ... < t_k < ... < t_{K-1}$, giving access to an empirical estimator of probability density $p_{t_k}(\mathbf{x})$ at time $t_k$,
$
{p}_{t_k} (\mathbf{x}) \approx  1/n_k \sum_{i = 1}^{n_k} \delta(\mathbf{x} - \mathbf{x}_{i,t_k}). \label{density}
$
The objective is then to infer the latent stochastic process that interpolates these observed marginals. Probability Flow Matching (PFM) does this within a simulation-free, continuous-time framework that learns stochastic processes of arbitrary complexity directly from multi-marginal data, while naturally incorporating growth dynamics. It builds on the Flow Matching (FM) framework \cite{lipman2022flow}, which learns velocity fields that transport a simple base distribution (e.g., a standard normal $p_0$) to a complex target distribution $p_1$ \cite{lipman2022flow}, bypassing repeated ODE solves, explicit likelihood evaluations, and Wasserstein computations, while retaining the advantages of continuous-time formulations. \\

To account for changes in cell abundance through proliferation and death within our PFM framework, we model cell population dynamics by introducing an auxiliary conditional mass density $m_t$, whose evolution governs growth and death rates in phase space (Appendix~\ref{App:unbalanced_TE}). Then, the empirical estimator of the probability density in the unbalanced setting at time $t_k$ can be written as,
\begin{equation}\label{eq:ub_prob}
    p_{t_k}(\mathbf{x}) \approx \sum_{i=1}^{n_k} m_{t_k}(\mathbf{x}_{i,t_k}) \delta (\mathbf{x} - \mathbf{x}_{i,t_k}).
\end{equation}
Based on this, we derive (see Appendix~\ref{sec:transport_equation}) the conditional flow matching objective in the unbalanced setting as:
\begin{equation}\label{eq:CFMg}
\mathcal{L}_{\text{CFM}} \equiv  \mathbb{E}_{t,q(\mathbf{z})} \Big[ \frac{m_t(\mathbf{z})}{M_t}\mathbb{E}_{p_t( \mathbf{x} \vert \mathbf{z})} \left\| \mathbf{u}_t(\mathbf{x} \vert \mathbf{z}) - \mathbf{u}_t^{\theta}(\mathbf{x}) \right\|^2 \Big ].
\end{equation}
where $\mathbf{z}$ is the conditioning variable, $q(\mathbf{z})$ is some prescribed distribution over the conditioning variable $\mathbf{z}$, $p_t(\mathbf{x}\vert \mathbf{z})$ the associated conditional path, $\mathbf{u}_t(\mathbf{x}\vert \mathbf{z})$ the corresponding conditional velocity, and $\mathbf{u}_t^\theta(\mathbf{x})$ the regressed velocity field. The normalization $M_t = \int p_t(\mathbf{x})d\mathbf{x}$ gives the total cell abundance at time $t$, and equals $1$ in the balanced setting. To generate a smooth and continuous velocity field that interpolates across the multi-marginal distributions, we set $q(\mathbf{z})$ to the multi-marginal optimal transport coupling $\pi^*$ (\textcolor{blue}{Material and Methods}) and use Chebyshev interpolants to construct the Gaussian conditional path distributions $p_t(\mathbf{x}\vert \mathbf{z})$. \\

\noindent \paragraph*{Score estimation:} As in PFI \cite{maddu2025learning}, we first estimate the time-dependent score function $\nabla \log p_t(\mathbf{x})$ from empirical samples at each measured time point. We use de-noising score matching techniques that provide fast and efficient, high-dimensional score estimation \cite{song2020sliced, song2020score}. Although one could jointly train the score and force fields \cite{tong2024simulation}, this leads to potential identifiability issues, as errors in the score can be absorbed by the drift and diffusion terms. We therefore find it more principled to estimate the score model independently and using it within the CFM regression objective.\\


In summary, our approach proceeds in three steps: (i) estimate the time-dependent score function from snapshot data, (ii) construct smooth conditional paths across marginals, and (iii) regress the probability flow velocity field as described in Eq.~(\ref{eq:CFMg}); see Algorithm.~1 for the full procedure.

\subsection{PFM enables accurate and scalable inference of stochastic processes from multi-marginal data}

Extending Flow Matching to the multi-marginal setting necessitates a careful parameterization of the conditional path $p_t(\mathbf{x}\vert\mathbf{z})$, because this parameterization determines the conditional velocity guiding the regression. In standard Flow Matching, the conditional path distribution is chosen as a Brownian bridge with diffusion scale $\sigma$, interpolating between samples
 $\mathbf{x}_0$ from the base distribution and $\mathbf{x}_1$ from the target distribution \cite{lipman2022flow}.
 

A natural extension to the multi-marginal setting is to allow the conditional mean path $\bm{\mu}_t$ to evolve piecewise linearly in time, following the same principle as standard Flow Matching between two distributions. This, however, yields piecewise constant velocities with discontinuities at the data points, introducing artifacts in the regressed velocity field. Such discontinuities hinder the inference of stochastic processes, where the underlying characteristic flow velocity is typically smooth and continuous in time. In PFM, we address this  issue by parameterizing the mean path $\bm{\mu}_t$ of the Brownian bridges using interpolating Chebyshev polynomials $Q(\mathbf{x}_0,\dots,\mathbf{x}_{K-1})$, i.e. $p_t(\mathbf{x} \vert \mathbf{z}) = \mathcal{N}(\mathbf{x} \vert  Q_t(\mathbf{z}), \sigma^2)$ which yield globally smooth time derivatives, in contrast to linear interpolants. In contrast to local polynomial interpolants such as cubic-spline, Chebyshev interpolation exhibits faster convergence in practice, particularly advantageous in the multi-marginal setting.

To illustrate this, we use a 2D Ornstein–Uhlenbeck (OU) process, for which force and diffusion operate on comparable time-scales, to infer a linear drift matrix $\mathbf{B}$ from time-resolved marginal data. The multi-marginal data is generated at time points $(t_1, t_2, t_3, t_4)$ from simulating the OU process with additive noise with constant strength (Fig.~\ref{fig:ou-interpolation}a). We then apply PFM to this dataset to recover $\mathbf{B}$, using different choices of conditional paths based on different parameterizations of $\bm{\mu}_t$: Linear, Cubic-Spline, and Chebyshev. The reconstructed probability flow velocities $\hat{\mathbf{u}}_t(\mathbf{x})$ (red arrows) are overlaid over the analytically evaluated probability flow velocity fields (yellow arrows, Fig.~\ref{fig:ou-interpolation}b) for different interpolation schemes. Among the methods, Chebyshev interpolation yields the most accurate estimation of the drift matrix, achieving nearly an order-of-magnitude reduction in RMSE compared to linear and cubic-spline interpolants. This underscores the importance of the conditional path parameterization in accurately interpolating between marginals and recovering the underlying dynamics.\\

To evaluate the accuracy and convergence of PFM and to quantify the effect of different interpolants in high-dimensional nonlinear settings, we consider a multistable system in $\mathbb{R}^d$ governed by a Waddington-like potential landscape (Fig.~\ref{fig:Bifur-interpolation}a) with additive noise of constant strength. We generate synthetic data in dimensions $d \in \{10,20,30,40,50\}$ and generate multi-marginal snapshot data at $K=5$ time points. PFM is then applied to reconstruct the underlying force field using three parameterizations of the mean path $\bm{\mu}_t$. As shown in Fig.~\ref{fig:Bifur-interpolation}b, the RMSE varies with both dimensionality $d$ and sample size $n$, with Chebyshev interpolation exhibiting superior accuracy in the high-dimensional, low-sample regime. It also shows improved convergence compared to the linear and cubic interpolants when increasing the number of quadrature points used to approximate the time-dependent expectation $\mathbb{E}_t[\cdot]$ in Eq.~(\ref{eq:CFMg}) (Fig.~\ref{fig:Bifur-interpolation}c). In summary, PFM with conditional paths parameterized by Chebyshev polynomial interpolants enables robust, accurate, and scalable inference of stochastic processes from multi-marginal datasets. In the next section, we apply PFM to time-resolved scRNA-seq datasets to infer cell differentiation dynamics.





\begin{figure*}
  \centering
  \includegraphics[width = 17.5cm]{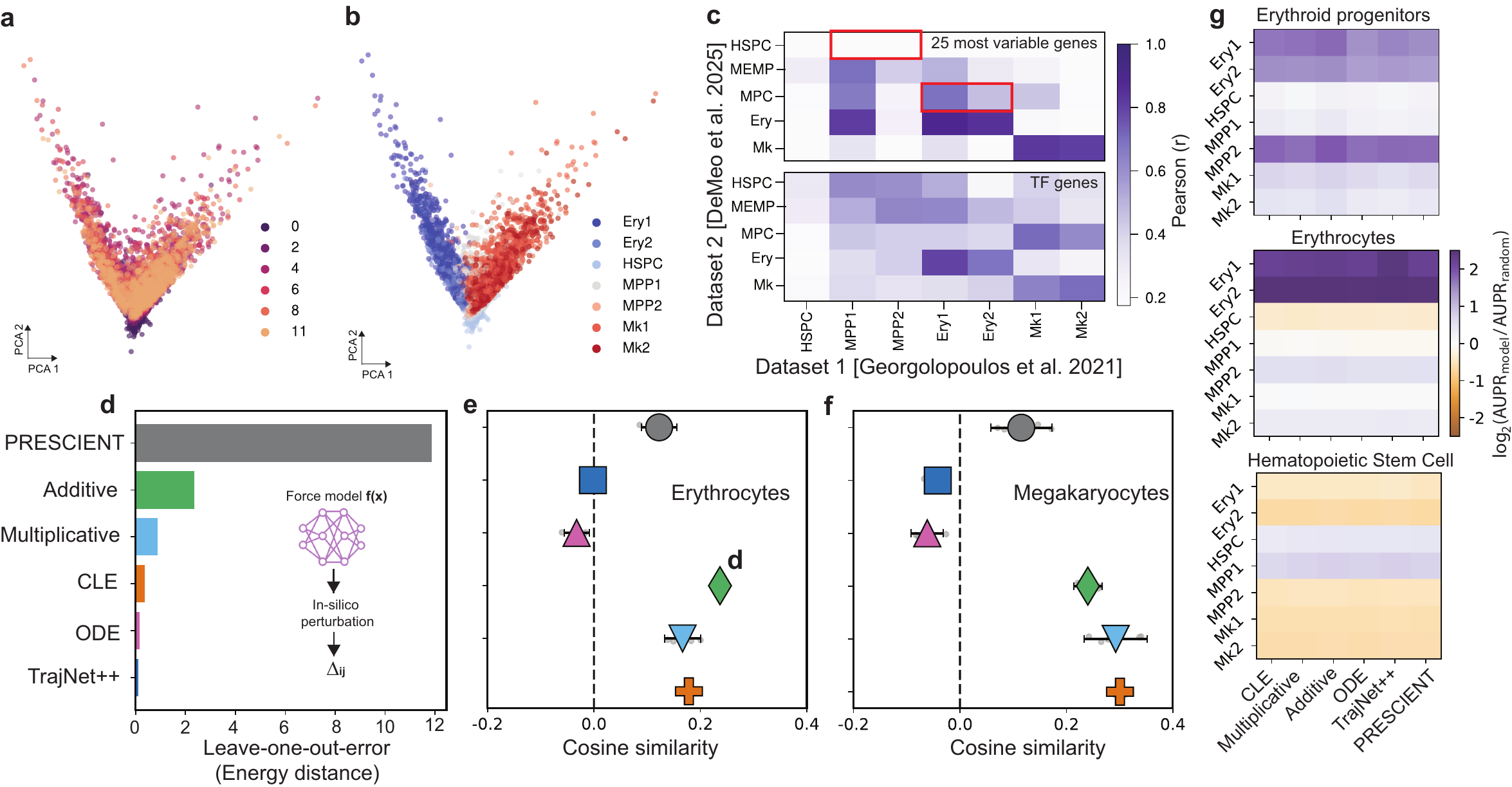}
\caption{
    \textbf{Inferring Gene Regulatory Interactions Using PFM:} Principal component projections of mRNA counts for 24 key transcription factors (TFs) measured during ex vivo hematopoiesis, colored by day \textbf{(a)} and cell type \textbf{(b)}.  \textbf{(c)} Cross-dataset comparison of gene–gene correlations during CD34+ HSC differentiation. The top heatmap shows correlations among the 25 most variable genes, and the bottom heatmap shows correlations among similar number of TF genes. Red boxes highlight regions with incorrectly inferred correlations that arise only in the most variable gene set. Expanded forms of the cell-type abbreviations are provided in Table~\ref{tab:cell_types} in the Appendix. \textbf{(d)} Interpolation accuracy, quantified by leave-one-out energy distance, for different dynamical models fitted to the marginal data using PFM. (Inset) In-silico perturbation of the inferred force field $\mathbf{f(x)}$ to estimate the regulatory response matrix $\Delta_{ij}$. \textbf{(e,f)} Cosine similarity between the inferred signed gene-regulatory interactions, indicating activation or inhibition, and experimentally validated networks for erythroid and megakaryocytic lineages, respectively. \textbf{(g)} Area under the precision–recall curve (AUPR) comparing cell-type-specific (titles) inferred gene-regulatory networks against corresponding ChIP-seq datasets across different methods. 
    }
    \label{fig:HSC_GRN}
\end{figure*}

\subsection{PFM enables accurate inference of differentiation dynamics in the Hematopoietic system}

Hematopoiesis is a canonical example of cell differentiation where mature blood cells of the myeloid and lymphoid lineage emerge from Hematopoietic stem cells (HSCs). Increasing evidence suggests that Hematopoietic differentiation is best understood as a continuous dynamical stochastic process, characterized by transient phases of indeterminate or mixed-lineage states \cite{dussiau2022hematopoietic}. Extensive single-cell transcriptomic (scRNA-seq) and chromatin accessibility (scATAC-seq) datasets now capture the transcriptional and regulatory dynamics of Hematopoietic differentiation. These measurements jointly profile thousands of genes and regulatory loci, defining a high-dimensional landscape where stem and progenitor cells commit to distinct lineages. 

Given such time-resolved multi-marginal data, the objective is to infer an underlying stochastic dynamics of differentiation. However, this is an inherently ill-posed inverse problem, with identifiability being a long-standing challenge in single-cell omics analysis \cite{weber2018identification}. This is reflected in the fact that many combinations of force and diffusion can interpolate the data equally well. For example, diffusion-less non-autonomous force models and autonomous Chemical Langevin models have been shown to have comparable interpolative accuracy \cite{maddu2025learning, tong2020trajectorynet}. The PFM provides a flexible framework to infer dynamical models of varying force and diffusion form that reconstructs the differentiation dynamics in high-dimensional TF space.

\begin{figure*}
  \centering
  \includegraphics[width = 17.5cm]{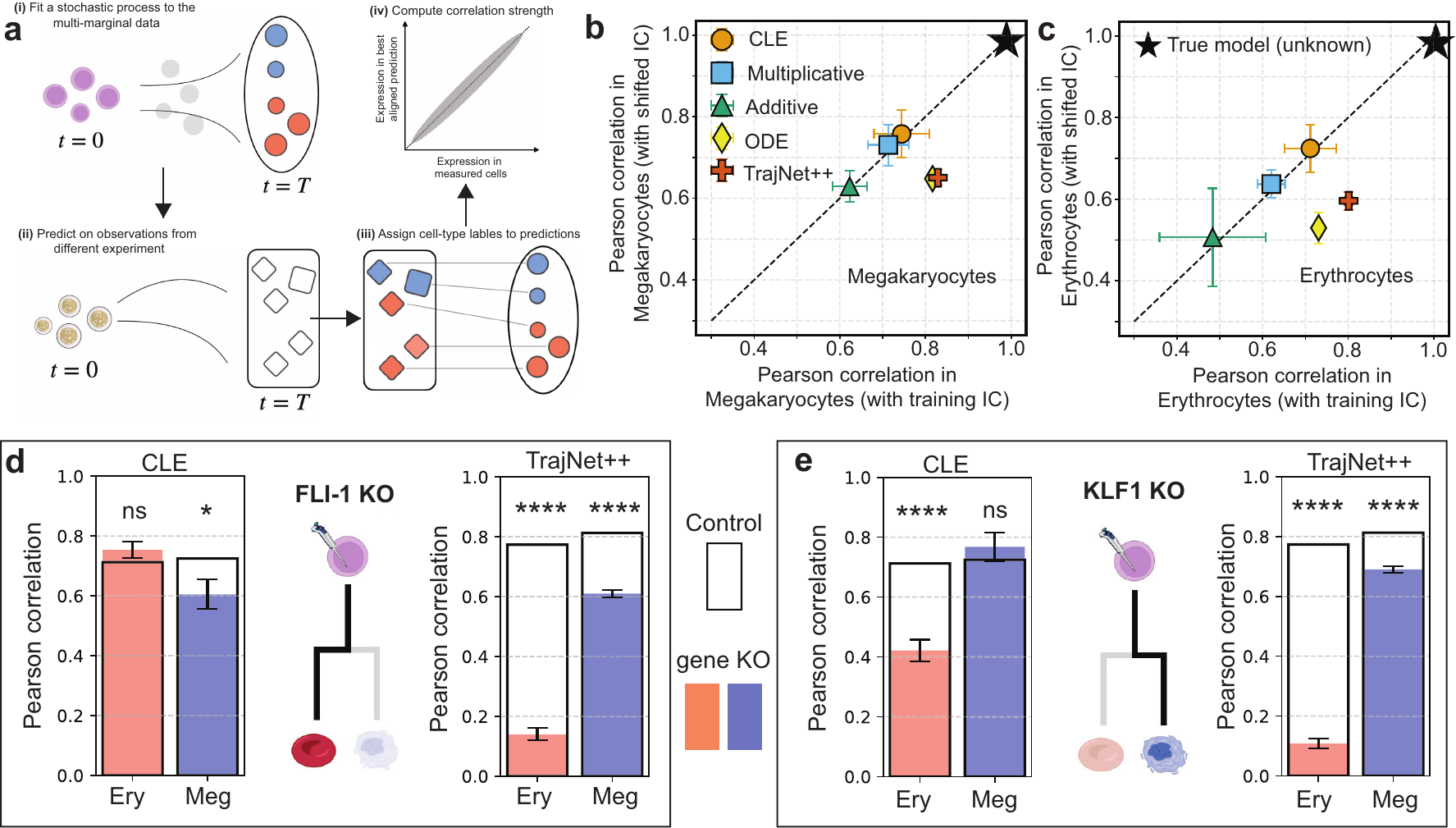}
\caption{
    \textbf{Evaluating Generalization Across Initial Conditions:}
     (a) Schematic of the validation strategy: models trained on differentiation trajectories originating from stem or progenitor populations are tested by predicting terminal gene-expression states from unseen initial conditions. Predicted and experimentally measured terminal cell distributions are first aligned using optimal transport (OT), and model performance is quantified by the correlation between aligned expression profiles. (b,c) Comparison of prediction accuracy across models for (b) megakaryocytic and (c) erythroid differentiation, showing the correlation between training and testing performance when initialized from HSC/HSPC states. The star symbol denotes perfect generalization. (d) Each panel compares the Pearson correlation of knockout (KO) populations with control Erythrocytes and Megakaryocytes, shown here in black outline. FLI1-KO denotes knockout of the FLI1 gene, which disrupts megakaryocyte formation while preserving erythroid identity. KLF1-KO denotes knockout of the KLF1 gene, which impairs erythroid differentiation but leaves the megakaryocytic lineage largely intact.
    }
    \label{fig:generalization}
\end{figure*}

We demonstrate this flexibility by applying PFM to time-resolved mRNA count data collected during \textit{ex vivo} Hematopoiesis \cite{georgolopoulos2021discrete}. The dataset captures cells undergoing a 12-day \textit{ex vivo} differentiation of CD34$++$
hematopoietic stem and progenitor cells (HSPCs) into erythroid and megakaryocytic lineages. Single-cell RNA-seq measurements were obtained on days 2, 4, 6, 8, and 11 following induction. Despite the high-dimensionality of the data, differentiation is largely governed by a small set of transcription factors (TFs) whose combinatorial activity drives lineage-specific programs \cite{reiter2017combinatorial,gillespie2020absolute}. We model differentiation within this reduced TF space, which preserves the essential regulatory architecture while remaining computationally tractable and interpretable. As shown in Fig.~\ref{fig:HSC_GRN}a,b, the PCA projection of mRNA counts for a set of 24 transcription factor (TF) genes accurately captures both the temporal progression of differentiation and the branching corresponding to erythropoiesis and megakaryopoiesis. These genes encode a minimal set of biophysically active regulators that modulate transcriptional activity and control differentiation, and their mRNA count distributions form a cell-type-specific signature that is conserved across datasets. This is illustrated in Fig.~\ref{fig:HSC_GRN}c, where we correlate these distributions across cell types between two different datasets of CD34+ HSC differentiation \cite{demeo2025active, georgolopoulos2021discrete}. We find high Pearson correlations for matching cell types, from progenitor to differentiated stages. By contrast, the distributions of non-TF genes, such as the 24 most highly variable genes of the \emph{ex vivo} hematopoiesis dataset \cite{georgolopoulos2021discrete}, while well conserved for differentiated cell types, do not distinguish progenitor cells as effectively.
\\

We fit the resulting multi-marginal distributions using PFM with different formulations of the force and diffusion terms.  As shown in Fig.~\ref{fig:HSC_GRN}d, the diffusion-free TrajectoryNet model trained with PFM, which employs a non-autonomous force field, achieves the smallest interpolation error (leave-one-out-error), followed by the ODE and Langevin models. In contrast, the PRESCIENT model, which enforces a strictly conservative force, overconstrains the dynamics and fails to capture the branching structure and temporal progression of differentiation, resulting in substantially higher interpolation error. We observe a similar trend in other Hematopoietic datasets with more than two terminal states (Fig.~\ref{fig:interpolation_error}). Together, these results highlight that a sufficiently expressive flow model can interpolate the data accurately, even when its underlying dynamics are not biophysically consistent.\\

\subsection{Biophysically consistent models lead to accurate gene-regulatory network inference}

Next, we ask whether high interpolation accuracy is sufficient to recover meaningful gene-regulatory interactions. To test this, we examine the causal relationship between genes via \textit{in silico} perturbations of the inferred force field $\mathbf{f}$ described in \cite{maddu2025learning}. Since the force model is explicitly a function of the state $\mathbf{x}$, this enables estimation of the local regulatory matrix (the Jacobian/response matrix) for each cell, thereby capturing the plasticity of gene-regulatory programs that continuously evolve across conditions and cell states. This is a fundamental feature of biological regulation that static network models overlook. \\


We illustrate this point by comparing the inferred networks for each cell type against corresponding chromatin immunoprecipitation sequencing (ChIP-seq) data. ChIP-seq reports TF–target binding interactions and therefore captures the structure of the underlying unsigned directed regulatory graph. Across all models, the AUPR computed against ChIP-seq data is highest for the matching cell type. For example, the network inferred for erythrocytes aligns most closely with the regulatory graph estimated from ChIP seq data of erythrocytes (Fig.~\ref{fig:HSC_GRN}g). These results highlight that regulatory interactions vary across cell types, reflecting context-specific gene regulation. At the same time, comparable AUPR values across models indicate that each captures the dominant directed dependencies among genes in a cell-type-specific manner (Fig.~\ref{fig:HSC_GRN}g).\\

Finally, we assessed whether the models can be distinguished by their ability to recover signed directed interactions, which encode not only direction but also the nature of regulation (activation or inhibition). For this, we design a similarity metric that accounts for both directionality and the sign of edges \cite{maddu2025learning}. Consistent with our previous findings \cite{maddu2025learning}, biophysically consistent models that account for intrinsic stochasticity, such as the Chemical Langevin model, enable a clearer separation of regulatory signals from noise. This accurately captures the causal interactions among key transcription factors driving erythrocyte–megakaryocyte fate bifurcation (Fig.~\ref{fig:HSC_GRN}e,f). This is substantiated by the fact that TF mRNA levels are often low relative to other genes \cite{weinreb2020lineage,gillespie2020absolute} (Fig.~\ref{fig:tfs_rest}), which potentially amplifies intrinsic stochastic effects and reshapes the underlying epigenetic landscape \cite{coomer2022noise}. In contrast, deterministic models such as ODE and TrajectoryNet, although accurate at interpolating data, fail to recover the correct regulatory dynamics because they attempt to model a noise-distorted landscape deterministically.








\subsection{Biophysically consistent models generalize well to unseen conditions}

In the previous section, we showed that accurate data interpolation does not guarantee correct recovery of gene-regulatory interactions. Two models with distinct force fields $\mathbf{f}_1(\mathbf{x},t)$ and $\mathbf{f}_2(\mathbf{x},t)$ can fit the data equally well and even share the same fixed points. For instance, both the Langevin and ODE models reproduce the terminal cell states, yet only the biophysically consistent Chemical Langevin model captures meaningful regulatory structure, indicating that their underlying force fields differ. In many biological systems, however, a ground-truth regulatory network is not available for direct validation. In such cases, one way to assess whether a model has learned the correct regulatory interactions is to evaluate its ability to generalize to unseen initial conditions.\\



We test this hypothesis by fitting models with specified force and diffusion forms to marginal data and assessing their generalization under shifted initial conditions. Specifically, we train on an \textit{ex vivo} differentiation dataset where samples at early time $t_0$ are enriched in HSPCs that differentiate into two terminal states: erythrocytes and megakaryocytes. A model that has correctly inferred the underlying regulatory interactions among key transcription factors should recover the same fixed points (terminal states) when initialized with HSCs, whose progeny give rise to the HSPCs. To test this, we align the model’s predicted terminal cells with labeled terminal cells from the training data by matching each cell in the predicted distribution to its most similar experimental counterpart (see Materials and Methods). Cell-type specific Pearson correlations between the paired cell's gene expression profiles quantify how well the model reproduces the correct terminal identities (See Fig.~\ref{fig:generalization}).\\

On the training data, methods that accurately interpolate the multi-marginal distributions, such as TrajectoryNet and ODE, naturally achieve high correlations, as reproducing the observed fixed points is intrinsic to their interpolation objective (Fig.~\ref{fig:HSC_GRN}d). However, when evaluated on data with shifted initial conditions (\textit{in vitro} CITE-seq data \cite{demeo2025active}), these correlations drop sharply (Fig.~\ref{fig:generalization}b,c), indicating a failure to recover the same attractors. Among all tested models, the CLE best recovers the terminal expression patterns corresponding to the erythroid and megakaryocytic fates, even when evaluated on unseen initial conditions not present during training. This trend persists when the training and test datasets are swapped (Fig.~\ref{fig:kaggle_generalization}). Overall, we found that models that accounted for some form of stochasticity retained similar correlations in both the training and test datasets.\\

We next assess whether the learned models could reproduce lineage-specific outcomes under \textit{in-silico} perturbations that mimic transcription-factor knockouts (KOs). To simulate these perturbations, we set the expression of selected transcription factors to zero and quantified how each model’s predicted terminal distribution aligned with the control. Specifically, we simulated knockouts of KLF1 (Kr\"{u}ppel-like factor 1), a key Erythroid transcription factor, and FLI1 (Friend leukemia integration 1), a transcription factor essential for Megakaryocytic differentiation. Experimentally, KLF1-KO impairs erythroid maturation while sparing megakaryocytes, whereas FLI1-KO blocks megakaryocytic differentiation but leaves the erythroid lineage intact.  Deterministic models such as TrajectoryNet+ and ODE failed to reproduce these lineage-specific effects, predicting loss of both lineages under FLI1-KO (Fig.~\ref{fig:KO}). In contrast, stochastic formulations, particularly the biophysically constrained Langevin and CLE models, correctly reproduce the asymmetric responses: KLF1-KO selectively disrupts erythroid differentiation while preserving megakaryocytes $(p=0.11)$, and FLI1-KO impairs megakaryocytic differentiation but leaves the erythroid lineage intact $(p=0.28)$. Models incorporating stochasticity and biophysical constraints thus generalize more faithfully to unseen perturbations and better capture the mechanistic basis of cell-fate specification.

\subsection{PFM enables inferring stochastic process along with growth dynamics}

\begin{figure*}
  \centering
  \includegraphics[width = 17.cm]{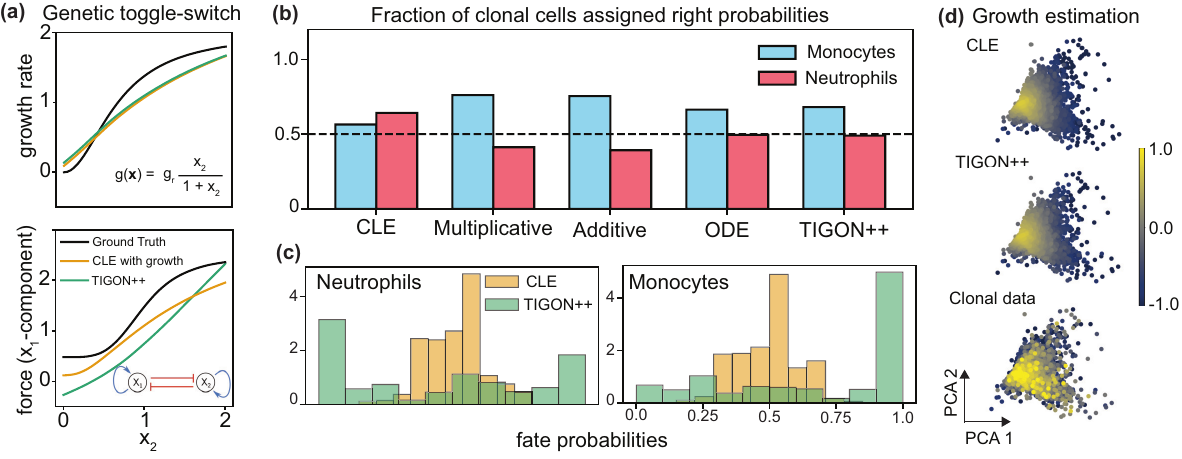}
\caption{
    \textbf{Inferring both growth and gene-regulatory dynamics using unbalanced PFM:} 
\textbf{(a)} Growth–rate and force–field estimation in a two-gene genetic toggle-switch model, comparing ground-truth dynamics with CLE-based inference and TIGON++.
\textbf{(b)} Fraction of clonal cells correctly assigned to their true fate across different inference models on the \textit{in vitro} Hematopoietic dataset. \textbf{(c)} Distributions of inferred fate probabilities for neutrophils and monocytes, highlighting differences in fate-decision behavior between CLE (stochastic) and TIGON++ (deterministic) formulations.
\textbf{(d)} Growth-rate estimations reconstructed from CLE, TIGON++, and clonal-lineage information.
    }
    \label{fig:growth}
\end{figure*}

During differentiation, cells not only change state but also undergo continuous proliferation and death. Existing approaches, such as UPFI \cite{zhang2025inferring} and TIGON \cite{sha2024reconstructing}, can infer both drift and population-level birth–death dynamics, but both are simulation-based, which requires repeated simulation and optimization of ODEs in high dimensions. Within our simulation-free PFM framework, we can jointly infer both the gene-regulatory dynamics $(\mathbf{h(x))}$ and the growth dynamics $(g(\mathbf{x}))$ through an additional objective that penalizes deviations from the observed population imbalances:
\begin{equation}\label{eq:CFM_growth}
    \hat{\mathbf{h}}, \hat{g} = \arg \min_{\mathbf{f},{g}} \mathcal{L}_{\text{CFM}} + \lambda_g \sum_k \vert M_{t_k} - \widehat{M}_{t_k}\vert
\end{equation}
where the estimated total mass is $\widehat{M}_{t_k} =  \sum_{i=1}^{n_k} m(\mathbf{x}_{i,t_k})$ and the local mass $m_t$ evolves along the probability path $\bm{\mu}_t$ according to $\dot{m}(\bm{\mu}_t) = m(\bm{\mu}_t) g(\bm{\mu}_t)$, where  $g(\bm{\mu}_t)$ denotes the growth rate along the trajectory. The regularization strength $\lambda_g$ is adaptively tuned following the weighting strategy described in \cite{maddu2022inverse}. In {Fig.~\ref{fig:growth}a}, we demonstrate how the PFM framework can jointly infer both gene-regulatory and growth dynamics using a genetic toggle switch as an illustrative example. Notably, the accuracy of the inferred growth dynamics appears to show limited sensitivity to the specific parameterization of drift and diffusion underlying the stochastic process.\\

We then applied PFM to time-resolved mRNA count data from the hematopoietic system, where progenitor cells are progressively differentiated into Neutrophils and Monocytes. The dataset captures cells undergoing a $6$-day \textit{in vitro} differentiation while also characterized by pronounced cell proliferation and death dynamics \cite{weinreb2020lineage}. Additionally, the dataset also provides lineage relationships between cells across multiple time points, which serve as first-order ground truth for evaluating the accuracy of the inferred growth dynamics $g$. To account for lineage-specific growth dynamics, we parameterize our growth dynamics as a function of the transcriptome state, i.e., $g(\mathbf{x})$, and parameterize using a feed-forward neural network. We then fit the corresponding multi-marginal distributions by minimizing the objective in Eq.~(\ref{eq:CFM_growth}), incorporating both force and diffusion terms as described in the previous section, along with the growth dynamics. \\

Consistent with our earlier findings, purely data-driven models such as TIGON++ and ODE achieve the lowest interpolation error (Fig.~\ref{fig:interpolation_error}a). However, when evaluated on their ability to predict cell fates, the underlying choice of stochastic dynamics has a pronounced effect on the inferred trajectories, and consequently on fate predictions (Fig.~\ref{fig:growth}b). Notably, only the CLE model correctly predicts the fate of more than $50\%$ (random fate assignment) of clonally tracked cells in both the neutrophil and monocyte populations. Not surprisingly, the distributions of inferred fate probabilities reveal distinct mechanisms of fate specification: TIGON++ produces more deterministic assignments, whereas CLE yields probabilistic, with slight bias to make fate decisions (Fig.~\ref{fig:growth}c). Finally, the growth-rate estimates $g(\mathbf{x})$ on day$\sim 4$ (Fig.~\ref{fig:growth}d) show no substantial differences across models, similar to our observations in the two-gene toggle-switch example.


\subsection{Discussion}
Here we introduce PFM, a scalable framework for inferring biophysically-consistent stochastic dynamics directly from multi-marginal cross-sectional single-cell data. PFM leverages recent algorithmic advances in generative modeling, like score- and flow matching, to fit a Fokker-Planck equation to distributional datasets. We also introduce methodological advances that extend FM to the multi-marginal setting: (i) a closed-form conditional flow-matching objective for inferring stochastic processes with an arbitrary form of stochasticity; (ii) a Chebyshev-polynomial based temporal interpolation scheme that parametrizes conditional probability flow paths and yields smooth conditional velocity estimates, outperforming linear and cubic spline interpolants in both accuracy and convergence; and (iii) an extension of flow matching to unbalanced settings to captures cell proliferation and death dynamics. \\

Another key aspect of our approach is the choice of representation of transcriptional state of the cell. Rather than operating in generic low-dimensional embeddings that can obscure mechanistic structure, we model dynamics in a transcription factor (TF) gene space, which provides biophysically meaningful coordinates for gene regulation. This reduced space preserves the core gene-regulatory architecture while enabling mechanistic interpretability and computational tractability. Our results show that a small subset of TF mRNA distributions form conserved, cell-type-specific signatures across datasets, indicating that this representation suffices to capture differentiation dynamics. More broadly, this suggests that modeling in a TF gene space can provide a tractable and mechanistically interpretable representation for inferring gene-regulatory dynamics from high-dimensional single-cell data.\\


The PFM framework also provides an overarching setting to examine different generative models such as TrajectoryNet, PRESCIENT, and Schrödinger Bridge, allowing training both deterministic and stochastic flows within a single framework. In particular, this flexible framework allows inferring biophysical stochastic models that explicitly capture intrinsic noise, such as the Chemical Langevin Equation (CLE). Applying the PFM framework to the hematopoietic system, we demonstrated that CLE models inferred from multi-marginal scRNA-seq data accurately reconstructed observed distributions and recovered gene-regulatory interactions consistent with ChIP-seq evidence. They also capture signed, directed edges, encoding both the direction and nature of regulation (activation or inhibition), in agreement with known transcription-factor relationships. In contrast, deterministic continuous normalizing flows, whether autonomous (ODE-based) or non-autonomous (TrajectoryNet+), interpolated the data well but failed to recover accurate regulatory interactions.\\

We also evaluated how well the different models generalized to unseen conditions using data from an independent experiment. Stochastic formulations, particularly the CLE model, generalized robustly beyond their training regime, while their deterministic counterparts did not. When initialized from an unseen population of hematopoietic progenitors, the CLE model also accurately reconstructed terminal attractors corresponding to erythroid and megakaryocytic fates, maintaining high gene-expression correlations with experimental data. The same model also reproduced lineage-specific responses under \textit{in-silico} transcription-factor knockouts: KLF1-KO selectively disrupted erythroid differentiation, whereas FLI1-KO perturbed megakaryocytic fate. Deterministic models failed to generalize to unseen conditions and did not reproduce the lineage-specific effects of transcription-factor knockouts. While stochasticity generally helps explore the overall epigenetic landscape, we found that the molecular stochasticity generalizes better. This choice of stochasticity naturally reproduces the observation that the noise scale is dependent on the local geometry of the epigenetic landscape \cite{losick2008stochasticity, pujadas2012regulated}.\\

In an extension to an unbalanced setting, we applied PFM to hematopoietic data with pronounced population imbalances across time, inferring proliferation and death dynamics consistent with clonal lineage estimates. Beyond cell differentiation dynamics, PFM offers a broad general framework for inferring stochastic dynamics that drive diverse cellular processes. Its scalability, flexibility, and ease of implementation makes PFM a powerful approach for integrating first-principles biophysical modeling with time-resolved, high-dimensional single-cell multi-omics data, enabling quantitative modeling and control of stochastic dynamics underlying differentiation, development, and disease progression. 




\subsection{Limitations and future directions}

In this study, we focused on the well-characterized hematopoietic differentiation system and used prior knowledge of the key transcription factors to reduce the model’s dimensionality. However, identifying a minimal and relevant gene set for less-studied systems, such as neural or T-cell differentiation, remains an open challenge. During embryonic development, gene-regulatory programs are tightly controlled in space and time. This suggests that future extensions of PFM should integrate spatial context information from spatial omics with time-resolved single-cell data to capture how local microenvironments influence regulatory decisions. Another important direction is to incorporate clonal lineage information from longitudinal transcriptomic data \cite{weinreb2020lineage} or high-resolution trajectories of marker genes to constrain further and validate the inferred dynamics. 



Currently, PFM employs a simple feedforward neural network to parameterize the force field and score function. Future work could leverage more structured architectures, such as graph neural networks or masked sparse layers, to encode prior biological information, such as  
known regulatory interactions or excluded non-interacting gene pairs.

\section{Materials and Methods}


\subsection*{Sampling conditioning variable via Multi-Marginal Optimal Transport}
The choice of the conditional variable $\mathbf{z}$ is crucial for determining the conditional velocity $\mathbf{v}_t(\mathbf{x}\vert\mathbf{z})$, and consequently the inferred force field. We rely on the extension of the OT-CFM approach to sample the conditional variable from the multi-marginal OT coupling \( \pi^*  \). For a choice of the cost function that is pairwise additive, the MMOT problem reduces to a set of $K-1$ independent OT problems \cite{rohbeck2025modeling}. Thus, when the pairwise optimal transport plans \( \pi_k^* \) between consecutive time points are known, we approximate the full MMOT coupling by
\begin{equation}
    \pi^*(i_0, \dots, i_{K-1}) \propto \frac{\prod_{k=0}^{K-2} \pi_k^*(i_k, i_{k+1})}{\prod_{k=1}^{K-2} \mu_k(i_k)}.
\end{equation}
This not only enforces temporal marginal consistency but also provides a principled way to sample the most probable trajectory across time, represented by the tuple $\mathbf{z} = (\mathbf{x}_0, \mathbf{x}_1,\cdots,\mathbf{x}_{K-1})$.

\subsection*{Chemical Langevin description}
The mathematical framework to account for intrinsic molecular stochasticity is the Chemical Fokker-Planck Equation (CFPE) and its associated Chemical Langevin
Equation (CLE). In the PFM framework, this corresponds to the force and diffusion taking the form:
\begin{align}
    \mathbf{f(x)} &=  \mathbf{h}(\mathbf{x}) - \ell \mathbf{x}, \label{eq:force}\\
    \mathbf{D(x)} &= \frac{1}{2}\text{diag}( h_1(\mathbf{x})+\ell x_1,\cdots,h_d(\mathbf{x})+\ell x_d). \label{eq:diff}
\end{align}
Here, $h_i(\mathbf{x})$, \( 0 < h_i(\mathbf{x}) < 1 \) is the activation function, and \( \ell \) is the degradation rate of mRNA molecules. 

\subsection*{Chebyshev Interpolation with Regularization}

To reconstruct smooth trajectories from discrete observations, we employed Chebyshev polynomial interpolation to parameterize our conditional paths. Given time points \( t \in [a, b] \), we map them to the canonical interval \( s \in [-1, 1] \) via the affine transformation
\[
s = \frac{2t - (a + b)}{b - a}.
\]
 The interpolant is expressed as a truncated Chebyshev expansion of degree \( M \):
\[
Q(t) \approx \sum_{m=0}^{M} c_m T_m(s),
\]
where \( T_m(s) \) are Chebyshev polynomials of the first kind defined recursively by
\[
T_0(s) = 1, \quad T_1(s) = s, \quad T_{m+1}(s) = 2s T_m(s) - T_{m-1}(s).
\]

The coefficients \( \{ c_m \} \) are obtained by solving the regularized least-squares problem
\[
\min_{\mathbf{c}} \| \mathbf{V} \mathbf{c} - \mathbf{x} \|^2 + \lambda c^T \mathbf{R} c 
\]
where \( V \in \mathbb{R}^{K \times (M+1)} \) is the Chebyshev--Vandermonde matrix evaluated at rescaled time points, and the diagonal matrix \( \mathbf{R} \) specifies the type of regularization. For instance, an \( \ell_2 \) penalty sets all diagonal entries to one, while velocity- and curvature-based regularizations assign the \( m \)-th diagonal element as \( m^2 \) and \( m^4 \), respectively. The time derivatives of the interpolant are computed analytically using Chebyshev polynomials of the second kind \( U_m(s) \), with
\[
\frac{d}{dt} Q(t) = \frac{2}{b - a} \sum_{m=1}^{M} c_m \cdot m \cdot U_{m-1}(s).
\]

\subsection*{Data processing}
We use three publicly available single-cell RNA-seq datasets: two datasets of induced differentiation of human CD34$^+$ HSPCs \cite{georgolopoulos2021discrete}, one generated with 10X Chromium v1 \cite{georgolopoulos2021discrete} and one with CITE-seq \cite{demeo2025active}, and one dataset of induced mouse HSPC differentiation with LARRY barcoding for lineage tracing \cite{weinreb2020lineage}. To compute the UMAP embeddings shown in Fig.~\ref{fig:schematic}, we applied the following steps: (i) subsampling to 20000 cells; (ii) selecting 2000 highly variable genes using scanpy with the \texttt{flavor=seurat\_v3} method \cite{wolf2018scanpy}, (iii) library-size correction, (iv) variance stabilization via biwhitening \cite{chardes2025random} with mean centering, and (v) PCA with 40 components followed by UMAP on these components. To compute the PCA projection in Fig.~\ref{fig:HSC_GRN} we only discarded cells expressing less than $4$ TFs among the $24$ selected TFs. For both datasets, after subsampling, we discarded cells expressing fewer than 100 genes and genes expressed in fewer than $0.1\%$ of cells.

\subsection*{Verifying that TFs have low number of counts}
For each dataset, we checked whether transcription factor genes have statistically lower counts than random sets of non-TF genes. For this, we selected genes with total counts between $100$ and $3\times 10^6$. We computed the distribution of non-zero counts for transcription factors listed in \cite{ravasi2010atlas} (shown in red in Fig.~\ref{fig:HSC_GRN}) and compared it to the average non-zero count distribution for $100$ random sets of genes sampled with replacement among non-TF genes (shown in gray in Fig.~\ref{fig:HSC_GRN}). To assess statistical significance, we computed a one-sided p-value for each bin following the standard approach for permutation p-values \cite{phipsonpermutation}. For a given bin, this p-value simply counts how many of the $100$ random gene sets had more counts than the set of TF genes. This p-value is shown in inset in Fig.~\ref{fig:HSC_GRN} and demonstrates that TF genes are enriched in very low counts (close to one) and depleted in larger counts. In appendix we show similar results for the two other datasets \cite{demeo2025active, weinreb2020lineage}. The idea that TF genes have low counts may not hold for \emph{all} TF genes. For instance, in the human datasets \cite{georgolopoulos2021discrete, demeo2025active}, we noticed that some TF genes encoding proteins with ubiquitous roles in cell proliferation, transcription, and metabolism, such as BTF3, ENO1, NPM1, PTMA, YBX1, and other High Mobility Group protein-coding genes, have consistently large numbers of counts. Discarding these outlier TFs, along with any other gene expressed at comparable or higher levels, improves the statistical significance of the results.

\subsection*{Comparing count distributions between experiments}
\label{counts_pearson}
To assess the validity of transferring the model trained on the 10X Chromium V3 dataset \cite{georgolopoulos2021discrete} to the CITE-seq dataset \cite{demeo2025active}, we evaluate how the count distributions compare between pairs of cell types in each dataset. Let us denote the two datasets $A$ and $B$. Given a pair of cell types $(a \in A, b \in B)$, we compute the optimal transport plan, with respect to the squared Euclidean cost, matching the gene expression profiles $X^a$ of $2000$ cells of type $a$ to the gene expression profiles $Y^b$ of the same number of cells of type $b$. We then use this optimal transport plan to determine which cell in $b$ is most similar to each cell in $a$. Denoting $(\Pi_{ij})_{i \in a, j \in b}$ the transport plan, this is done by matching each cell $i \in a$ to the cell $k \in b$ defined by $k = \arg\max_j \Pi_{ij}$. We then compute the correlation between the gene expression profiles of each paired cell. More specifically, this is achieved by correlating simultaneously $X_{iq}$ with $Y_{kq}$ and $X_{ip}$ with $Y_{kp}$ for all genes $p \neq q$ and all cells $i$. This is the same assignment approach used to compute the Pearson correlation per cell type in Fig.~\ref{fig:generalization}.

\subsection*{Data and code availability}
The source code to reproduce the findings of this study are openly available at 
\hyperlink{PFI}{https://vchz.github.io/pfi/}, alongside the processed scRNA-seq datasets \hyperlink{PFI}{https://zenodo.org/records/19237708}

\subsection*{Acknowledgements}
We are grateful to graphic designer Lucy Reading-Ikkanda for her help with the figures in this study. We also thank Scott Weady, Gilles Francfort, Yuhai Tu, and Daniel Needleman for many insightful discussions.

\bibliographystyle{unsrt}
\bibliography{references}

\appendix
\renewcommand{\thefigure}{S\arabic{figure}}
\setcounter{figure}{0}

\onecolumngrid
\newpage
\section{Transport equation through conditional distributions}\label{sec:transport_equation}
In this section, we derive the continuity equation governing the time–evolving density \(p_t(\mathbf{x})\)induced by time-evolving conditional distribution constructed from multi-marginal distributions $q_0,q_1,\dots,q_{K-1}$ sampled at discrete times $t_k, 0\leq k\leq K-1$. We assume no cell proliferation or death, i.e., $g_t(\mathbf{x}) = 0$. We begin with a brief overview of the Gaussian conditional paths of \cite{lipman2022flow}, and then analyze their behavior in the limit of constant but vanishing variance.

\subsection{Gaussian flow with constant and vanishing variance}
For a fixed conditioning variable $\mathbf{z}$ defined with the tuple $(\mathbf{x}_0,\mathbf{x}_1,\dots,\mathbf{x}_{K-1})$, we recall the conditional Gaussian path considered by Lipman et al.~\cite{lipman2022flow} (Theorem~3) as
\[
p_t(\mathbf{x} \mid \mathbf{z})
= \mathcal{N}\!\big(\mathbf{x};\,\bm{\mu}_t(\mathbf{z}),\,\sigma_t^2(\mathbf{z})\big),
\qquad t\in[0,1],
\]
where \(\bm{\mu}_t(\mathbf{z})\) and \(\sigma_t(\mathbf{z})>0\) are differentiable in \(t\). The mean conditional paths $\bm{\mu}_t$ are constructed in a way to satisfy the boundary conditions. For instance, when $t= 0$, the mean conditional path $\bm{\mu}_{t_0} = \mathbf{x}_0$ and for $\sigma_t \rightarrow 0$, $p_{t_0}(\mathbf{x}\vert \mathbf{x}_0)$ is a concentrated Gaussian distribution centered at $\mathbf{x}_0$. Assume that Gaussian probability path \(p_t(\cdot\mid \mathbf{z})\) is realized as the pushforward of \(p_0(\cdot\mid \mathbf{z})\) by a flow map \(\psi_t\) generated by a
(time-dependent) vector field \(\bm{u}_t(\cdot\mid \mathbf{z})\). Then
Theorem~3 of~\cite{lipman2022flow} states that the vector field that defines $\psi_t$ has the form:
\begin{equation}
\bm{u}_t(\mathbf{x}\mid \mathbf{z})
    = \frac{\sigma_t'(\mathbf{z})}{\sigma_t(\mathbf{z})}\bigl(\mathbf{x} - \bm{\mu}_t(\mathbf{z})\bigr)
      + \bm{\mu}_t'(\mathbf{z}).
\label{eq:gaussian-velocity-general}
\end{equation}
\noindent We now record two simple consequences of the above relation in Eq.~(\ref{eq:gaussian-velocity-general}) that are useful for our analysis.

\begin{corollary}[Constant-variance Gaussian path and small-variance limit]
\label{cor:const-var-small-sigma}
Suppose that for each fixed \(\mathbf{z}\) the conditional variance is constant in time,
\[
    \sigma_t(\mathbf{z}) \equiv \sigma(\mathbf{z}) > 0, \qquad t \in [t_0,t_{K-1}].
\]
For each choice of \(\sigma(\mathbf{z})\) we obtain a Gaussian path
\[
    p_t^{(\sigma)}(\mathbf{x} \mid \mathbf{z})
    = \mathcal{N}\!\bigl(\mathbf{x};\,\bm{\mu}_t(\mathbf{z}),\sigma^2(\mathbf{z})\bigr), \qquad t\in[t_0,t_{K-1}].
\]

\begin{itemize}
 \item [(a)]For any fixed \(t\in[t_0,t_{K-1}]\) and \(\mathbf{z}\), as the (time-independent) variance
      parameter \(\sigma(\mathbf{z})\) tends to zero, the path collapses to a moving
      point mass, in the sense that
      \[
          p_t^{(\sigma)}(\cdot \mid \mathbf{z})
          \;\Rightarrow\;
          \delta\bigl(\mathbf{x} - \bm{\mu}_t(\mathbf{z})\bigr)
          \qquad \text{as } \sigma(\mathbf{z}) \rightarrow 0,
      \]
      where \(\Rightarrow\) denotes weak convergence of probability measures.\\

 \item [(b)]For every fixed \(\sigma(\mathbf{z}) > 0\), the associated Gaussian path
      \(p_t^{(\sigma)}(\mathbf{x} \mid \mathbf{z})\) is generated by the conditional velocity that only depends on the mean conditional path $\bm{\mu}_t$
      \[
          \bm{u}_t^{(\sigma)}(\mathbf{x} \mid \mathbf{z})
          = \dot{\bm{\mu}}_t(\mathbf{z}), \qquad \mathbf{x} \in \mathbb{R}^d,
      \]
      obtained from Eq.~(\ref{eq:gaussian-velocity-general}) by observing that
      \(\sigma_t'(\mathbf{z}) = 0\) when \(\sigma_t(\mathbf{z})\equiv \sigma(\mathbf{z})\) is fixed.
      In particular, the velocity field does not depend on \(\sigma(\mathbf{z})\),
      so the limiting point mass \(\delta(\mathbf{x}-\bm{\mu}_t(\mathbf{z}))\) is transported along
      the trajectory \(\bm{\mu}_t(\mathbf{z})\) with velocity \(\dot{\bm{\mu}}_t(\mathbf{z})\).
\end{itemize}
\end{corollary}

\subsection{Transport equation in the balanced setting}
\noindent Using the conditional distributions in the constant and vanishing variance limit derived in Corollary \ref{cor:const-var-small-sigma}(a,b), we define the corresponding marginal distribution in $\mathbf{x}$ as
\begin{equation}
p_t(\mathbf{x})  := \int p_t(\mathbf{x}\mid\mathbf{z}) q(\mathbf{z})d\mathbf{z} 
 = \int \delta (\mathbf{x} - \bm{\mu}_t(\mathbf{z})) q(\mathbf{z})d\mathbf{z}\label{eq:pt-def}.
\end{equation} 
\paragraph{Boundary conditions.} This form also satisfies the appropriate boundary conditions given by the data at all times $t_k, \forall k \in [0,\dots,K-1]$. Since the mean conditional path evaluates to the samples, i.e., \(\bm{\mu}_{t_k} = \mathbf{x}_k\), following Eq.~(\ref{eq:pt-def}), the marginal distribution naturally satisfies the boundary condition
\[
    p_{k}(\mathbf{x})=q_k(\mathbf{x}), \quad \forall \: k \in [0,\dots,K-1].
\]

\paragraph{Time derivative of marginal distribution \(p_t(\mathbf{x})\):}
\noindent Differentiate Eq.~\eqref{eq:pt-def} w.r.t time and using the chain rule $\partial_t \delta(\mathbf{x}-\bm\mu_t)
        = - \nabla_\mathbf{x} \delta(\mathbf{x}-\bm\mu_t) \cdot \dot{\bm\mu}_t$, we get
\begin{align}
    \partial_t p_t(\mathbf{x})
        &= \int \partial_t \delta(\mathbf{x}-\bm{\mu}_t)\,
            q(\mathbf{z})\, d\mathbf{z}.
            \label{eq:pt-deriv1} \\
  & = -\nabla_\mathbf{x} \cdot 
         \int 
            \delta(\mathbf{x}-\bm\mu_t)\,\dot{\bm \mu}_t(\mathbf{z})\,
            q(\mathbf{z})\, d\mathbf{z} .
    \label{eq:pt-deriv2}
\end{align}

\paragraph{Transport velocity $\mathbf{u}_t(\mathbf{x})$: }
\noindent We can also define the transport velocity, by marginalizing over the conditional probability flow $p_t(\mathbf{x}\vert \mathbf{z})$ induced by the conditional velocity $\mathbf{u}_t(\mathbf{x}\vert \mathbf{z})$ defined as follows:
\begin{equation}
    \mathbf{u}_t(\mathbf{x})
    :=  \int \mathbf{u}_t(\mathbf{x}\mid\mathbf{z})  \, \frac{p_t(\mathbf{x}\mid \mathbf{z}) q(\mathbf{z})}{p_t(\mathbf{x})}\, d\mathbf{z} .
    \label{eq:cond-vel}
\end{equation}
%
\noindent Multiplying both side by $p_t(\mathbf{x})$, we get
\begin{align}
   \mathbf{u}_t(\mathbf{x}) p_t(\mathbf{x}) & =  \int \mathbf{u}_t(\mathbf{x}\mid \mathbf{z}) \, p_t(\mathbf{x}\mid \mathbf{z})q (\mathbf{z})\ d\mathbf{z}\\
    \mathbf{u}_t(\mathbf{x}) p_t(\mathbf{x})& =  \int \dot{\bm\mu}_t(\mathbf{z}) \, \delta (\mathbf{x} - \bm{\mu}_t(\mathbf{z}))q (\mathbf{z})\ d\mathbf{z}  \quad \quad \text{using corollary \ref{cor:const-var-small-sigma}} \label{eq:vp-identity}
\end{align}

\begin{proposition}[Transport equation]
Let \(p_t\) be defined by \eqref{eq:pt-def}.  
Then \(p_t\) satisfies the continuity equation
\begin{equation}
    \partial_t p_t(\mathbf{x})
        = - \nabla_\mathbf{x} \cdot \big( \mathbf{u}_t(\mathbf{x})\, p_t(\mathbf{x}) \big),
    \label{eq:transport}
\end{equation}
where the velocity field \(\mathbf{u}_t(\mathbf{x})\) is the marginal velocity given by Eq.~\ref{eq:cond-vel}.

\end{proposition}

\begin{proof}
Substituting the identity \eqref{eq:vp-identity} into \eqref{eq:pt-deriv2} gives \eqref{eq:transport}.
\end{proof}

\subsection{Transport equation in the unbalanced setting}\label{App:unbalanced_TE}
In this section, we derive the evolution equation for a time-dependent unbalanced density $\tilde{p}_t(\mathbf{x})$ induced by conditional measures constructed from multi-marginal distributions. The key difference from previous section is that marginal distribution is no longer normalized.\\

\noindent We first introduce a (strictly positive) mass field
$m_t>0$, which we evaluate along the conditional mean path via $m_t(\mathbf{z})$.  We then define the conditional path as \emph{mass-weighted kernel}
\begin{equation}
\tilde{p}_t(\mathbf{x}\mid\mathbf{z})
    := m_t\big(\mathbf{z}\big)\,
       \delta\big(\mathbf{x}-\bm{\mu}_t(\mathbf{z})\big),
\label{eq:mass-kernel}
\end{equation}
which plays the role of a ``conditional mass'' assigned to the position
$\mathbf{x}$ conditioned on $\mathbf{z}$. This is analagous to the conditional probability distribution $p_t^{(\sigma\rightarrow 0)}(\mathbf{x}\vert \mathbf{z})$ discussed in Corollary \ref{cor:const-var-small-sigma} in the balanced setting.\\

\noindent The corresponding (unnormalized) mass density is then
defined as
\begin{equation}
\tilde{p}_t(\mathbf{x})
    := \int \tilde{p}_t(\mathbf{x}\mid\mathbf{z})\,q(\mathbf{z})\,d\mathbf{z}
     = \int \delta\big(\mathbf{x}-\bm{\mu}_t(\mathbf{z})\big)\,
            m_t\big(\mathbf{z} \big)\,
            q(\mathbf{z})\,d\mathbf{z}.
\label{eq:pt-unbalanced-def}
\end{equation}
In general $M_t := \int \tilde{p}_t(\mathbf{x})\,d\mathbf{x} \neq 1$, so $\tilde{p}_t$ should be interpreted as a mass density rather than a probability density.
\\

\noindent \paragraph{Boundary conditions.}
At $t=0$ and $t=1$, using $\bm{\mu}_0(\mathbf{z})=\mathbf{x}_0$ and
$\bm{\mu}_1(\mathbf{z})=\mathbf{x}_1$ we obtain
\begin{align}
\tilde{p}_0(\mathbf{x})
    &= \int \delta(\mathbf{x}-\mathbf{x}_0)\,
         m_0(\mathbf{x}_0)\,q(\mathbf{x}_0) q(\mathbf{x}_1)\,
       d\mathbf{x}_0 d\mathbf{x}_1
     = m_0(\mathbf{x})\,q_0(\mathbf{x}),\\
\tilde{p}_1(\mathbf{x})
    &= \int \delta(\mathbf{x}-\mathbf{x}_1)\,
         m_1(\mathbf{x}_1)\,q(\mathbf{x}_0) q(\mathbf{x}_1)\,
       d\mathbf{x}_0 d\mathbf{x}_1 = m_1(\mathbf{x})\,q_1(\mathbf{x}),
\end{align}
so $\tilde{p}_t$ coincides with the observed data marginals multiplied by the
appropriate endpoint mass fields.\\

\paragraph{Time derivative of unbalanced \(\tilde{p}_t(\mathbf{x})\): } Differentiating \eqref{eq:pt-unbalanced-def} with respect to time and
using the chain rule yields
\begin{equation}
\partial_t \tilde{p}_t(\mathbf{x})
    = \int \partial_t\!\left[
            \delta\big(\mathbf{x}-\bm{\mu}_t(\mathbf{z})\big)\,
            m_t\big(\mathbf{z} \big)
        \right] q(\mathbf{z})\,d\mathbf{z}.
\label{eq:pt-deriv-unbalanced-1}
\end{equation}
We expand the time derivative as
\begin{equation}
\partial_t\!\left[
    \delta(\mathbf{x}-\bm{\mu}_t(\mathbf{z}))\,m_t(\mathbf{z})
\right]
= \underbrace{\partial_t \delta(\mathbf{x}-\bm{\mu}_t(\mathbf{z}))\,
               m_t(\mathbf{z})}_{\text{transport}}
  + \underbrace{\delta(\mathbf{x}-\bm{\mu}_t(\mathbf{z}))\,
               \partial_t m_t(\mathbf{z})}_{\text{growth/death}},
\end{equation}
Using the chain rule,
$
\partial_t \delta(\mathbf{x} - \bm{\mu}_t)
=
- \nabla_{\mathbf{x}} \delta(\mathbf{x} - \bm{\mu}_t)
\cdot \dot{\bm{\mu}}_t,
$
and the identity
$
\partial_t m_t(\mathbf{z})
=
m_t(\mathbf{z})\, \partial_t \log m_t(\mathbf{z}),
$
we obtain
\begin{align}
\partial_t \tilde{p}_t(\mathbf{x})
&= -\int
     \nabla_{\mathbf{x}}\delta(\mathbf{x}-\bm{\mu}_t(\mathbf{z}))\cdot
     \dot{\bm{\mu}}_t(\mathbf{z})\,
     m_t(\mathbf{z})\,q(\mathbf{z})\,d\mathbf{z}
   \nonumber\\[0.1cm]
&\quad + \int
     \delta(\mathbf{x}-\bm{\mu}_t(\mathbf{z}))\,
     m_t(\mathbf{z})\,
     \partial_t\log m_t(\mathbf{z})\,
     q(\mathbf{z})\,d\mathbf{z}.
\label{eq:pt-deriv-unbalanced-2}
\end{align}
Using integration by parts in $\mathbf{x}$ for the first term, we can
write this compactly as
\begin{equation}
\partial_t \tilde{p}_t(\mathbf{x})
  = -\nabla_{\mathbf{x}}\cdot
        \int \delta(\mathbf{x}-\bm{\mu}_t)\,
             \dot{\bm{\mu}}_t(\mathbf{z})\,
             m_t(\mathbf{z})\,q(\mathbf{z})\,d\mathbf{z}
    + \int \delta(\mathbf{x}-\bm{\mu}_t)\,
             m_t(\mathbf{z})\,
             \partial_t\log m_t(\mathbf{z})\,
             q(\mathbf{z})\,d\mathbf{z}.
\label{eq:pt-deriv-unbalanced-3}
\end{equation}

\paragraph{Transport velocity in the unbalanced setting: }  We can also define the transport velocity, by marginalizing over the mass-weighted kernel $\tilde{p}_t(\mathbf{x}\vert \mathbf{z})$ induced by the conditional velocity $\mathbf{u}_t(\mathbf{x}\vert \mathbf{z})$ defined as follows:
\begin{equation}
    \mathbf{u}_t(\mathbf{x})
    :=  \int \mathbf{u}_t(\mathbf{x}\mid\mathbf{z})  \, \frac{\tilde{p}_t(\mathbf{x}\mid \mathbf{z}) q(\mathbf{z})}{\tilde{p}_t(\mathbf{x})}\, d\mathbf{z} .
\label{eq:unbalanced-cond-vel}
\end{equation}

\noindent Multiplying both sides of \eqref{eq:unbalanced-cond-vel} by $\tilde{p}_t(\mathbf{x})$, we find
\begin{align}
\mathbf{u}_t(\mathbf{x})\,\tilde{p}_t(\mathbf{x})
 &= \int \mathbf{u}_t (\mathbf{x} \mid \mathbf{z})\tilde{p}_t(\mathbf{x} \mid \mathbf{z})\,q(\mathbf{z})\,d\mathbf{z}\\
\mathbf{u}_t(\mathbf{x})\,\tilde{p}_t(\mathbf{x})
 &= \int \dot{\bm{\mu}}_t(\mathbf{z})\,
        \delta(\mathbf{x}-\bm{\mu}_t(\mathbf{z}))\,
        m_t(\mathbf{z})\,q(\mathbf{z})\,d\mathbf{z}.
\label{eq:vp-identity-unbalanced}
\end{align}

\paragraph{Growth rate:} Similarly to the transport velocity, we define the local growth rate
\begin{equation}
g_t(\mathbf{x})
   := \mathbb{E}_{\tilde{p}_t(\mathbf{z}\mid\mathbf{x})}
        \big[\partial_t\log m_t(\mathbf{z}) \mid \mathbf{x}\big]
    = \int \partial_t\log m_t(\mathbf{z})\,
         \tilde{p}_t(\mathbf{z}\mid\mathbf{x})\,d\mathbf{z},
\label{eq:unbalanced-cond-growth}
\end{equation}
which, after multiplying by $p_t(\mathbf{x})$ and inserting the
definition of $p_t(\mathbf{z}\mid\mathbf{x})$, gives
\begin{align}
g_t(\mathbf{x})\,\tilde{p}_t(\mathbf{x})
 &= \int \partial_t\log m_t(\mathbf{z})\,
        \delta(\mathbf{x}-\bm{\mu}_t(\mathbf{z}))\,
        m_t(\mathbf{z})\,q(\mathbf{z})\,d\mathbf{z}.
\label{eq:gp-identity-unbalanced}
\end{align}

\begin{proposition}[Unbalanced transport equation]
Let $p_t$ be defined by \eqref{eq:pt-unbalanced-def}. Then $\tilde{p}_t$ satisfies
the unbalanced continuity equation
\begin{equation}
\partial_t \tilde{p}_t(\mathbf{x})
  = -\nabla_{\mathbf{x}}\cdot\big(\mathbf{u}_t(\mathbf{x})\,\tilde{p}_t(\mathbf{x})\big)
    + g_t(\mathbf{x})\,\tilde{p}_t(\mathbf{x}),
\label{eq:unbalanced-transport}
\end{equation}
where the velocity field $\mathbf{u}_t$ and growth rate $g_t$ are given by
the conditional expectations
\[
\mathbf{u}_t(\mathbf{x})
   = \mathbb{E}_{\tilde{p}_t(\mathbf{z}\mid\mathbf{x})}
        \big[\dot{\bm{\mu}}_t(\mathbf{z}) \mid \mathbf{x}\big],
\qquad
g_t(\mathbf{x})
   = \mathbb{E}_{\tilde{p}_t(\mathbf{z}\mid\mathbf{x})}
        \big[\partial_t\log m_t(\mathbf{z}) \mid \mathbf{x}\big].
\]
\end{proposition}

\begin{proof}
Substituting the identities
\eqref{eq:vp-identity-unbalanced} and \eqref{eq:gp-identity-unbalanced}
into \eqref{eq:pt-deriv-unbalanced-3} yields
\eqref{eq:unbalanced-transport}:
\[
\partial_t \tilde{p}_t(\mathbf{x})
  = -\nabla_{\mathbf{x}}\cdot\big(\mathbf{u}_t(\mathbf{x})\,\tilde{p}_t(\mathbf{x})\big)
    + g_t(\mathbf{x})\,\tilde{p}_t(\mathbf{x}).
\]
\end{proof}

\subsection{Derivation of the Unbalanced Flow Matching Objective}

\noindent Given a mass density path $\tilde{p}_t(\mathbf{x})$ and a corresponding vector field $\mathbf{u}_t(\mathbf{x})$, which generates $p_t(\mathbf{x})$, we define the Unbalanced Flow Matching (UFM) objective as,
\begin{equation}\label{eq:UFM}
    \mathcal L_{\mathrm{UFM}}(\theta)
= \mathbb E_{t\sim \mathcal{U}[t_0,t_{K-1}],\mathbf{x}\sim p_t(\mathbf{x})}\!\big[\|\mathbf{u}_t^\theta(\mathbf{x})-\mathbf{u}_t(\mathbf{x})\|^2\big]
= \mathbb{E}_t \Big [\frac{1}{M_t}\int_{\mathcal X}\tilde p_t(\mathbf{x})\,\|\mathbf{u}_t^\theta(\mathbf{x})-\mathbf{u}_t(\mathbf{x})\|^2\,d\mathbf{x}\Big].
\end{equation}

\noindent 

\begin{theorem}
Assume that the $p_t(\mathbf{x})>0$ for all $\mathbf{x}\in\mathbb R^d$ and $t\in[0,1]$.
Then, up to a constant independent of $\theta$, the unbalanced flow matching loss in Eq.~(\ref{eq:UFM})
is equivalent to the conditional objective
\[
\mathcal L_{\mathrm{UCFM}}(\theta)
=
\mathbb E_{t, q(\mathbf{z})}
\left[
\frac{m_t(\mathbf{z})}{M_t}
\,
\mathbb E_{p_t(\mathbf{x}\mid\mathbf{z})}
\|\mathbf{u}_t^\theta(\mathbf{x})-\dot{\bm{\mu}}_t(\mathbf{z})\|^2
\right].
\]
\noindent Hence, 
\[
\nabla_\theta \mathcal L_{\mathrm{UFM}}(\theta)
\equiv
\nabla_\theta \mathcal L_{\mathrm{UCFM}}(\theta).
\]
\end{theorem}

\begin{proof}
\noindent Using the bilinearity of the squared norm, we can decompose the Unbalanced Flow Matching objective as,
\[
\|\bm{u}_{t}^\theta(\mathbf{x})-\bm{u}_t(\mathbf{x})\|^2
=
\|\bm{u}_t^\theta(\mathbf{x})\|^2
-2\langle \bm{u}_t^\theta(\mathbf{x}),\bm{u}_t(\mathbf{x})\rangle
+
\|\bm{u}_t(\mathbf{x})\|^2 .
\]
\noindent We then evaluate each term and its expectation under $p_t(\mathbf{x})$. We start with the term 
\begin{align}
     \frac{1}{M_t}\int_{\mathcal X}\tilde p_t(\mathbf{x})\,\|\mathbf{u}_t^\theta(\mathbf{x}))\|^2\ d\mathbf{x} &= \frac{1}{M_t} \int_{\mathcal X} \int_{\mathcal Z}\,\|\mathbf{u}_t^\theta(\mathbf{x})\|^2\  \tilde p_t(\mathbf{x} \vert \mathbf{z}) q(\mathbf{z}) d\mathbf{x} d\mathbf{z} \\
     &= \frac{1}{M_t} \int_{\mathcal X}\ \int_{\mathcal Z}\,\|\mathbf{u}_t^\theta(\mathbf{x}))\|^2\  m_t(\mathbf{z}) p_t(\mathbf{x\vert \mathbf{z}}) q(\mathbf{z}) d\mathbf{x} d\mathbf{z}  \\
    & = \frac{1}{M_t}\mathbb{E}_{q(\mathbf{z})}[m_t(\mathbf{z}) \mathbb{E}_{p_t(\mathbf{x\vert \mathbf{z}})} \Vert \mathbf{u}_t^\theta(\mathbf{x})\Vert^2]
\end{align}

\noindent Next,
\begin{align}
     \frac{1}{M_t}\int_{\mathcal X}\tilde p_t(\mathbf{x})\,\langle \mathbf{u}_t^\theta(\mathbf{x}), \mathbf{u}_t(\mathbf{x}) \rangle  \ d\mathbf{x} &=      \frac{1}{M_t} \int_{\mathcal X} \tilde p_t(\mathbf{x})\, \left\langle \mathbf{u}_t^\theta(\mathbf{x}), \frac{\int \dot{\mathbf{\mu}}_t(\mathbf{z})\, \tilde{p}_t(\mathbf{x}\mid \mathbf{z})\, q(\mathbf{z})\, d\mathbf{z}} {\tilde{p}_t(\mathbf{x})}
\right\rangle \, d\mathbf{x} \\
     &=  \frac{1}{M_t} \int_{\mathcal X} \, \left\langle \mathbf{u}_t^\theta(\mathbf{x}),\int_{\mathcal Z} \dot{\mathbf{\mu}}_t(\mathbf{z})\, \tilde{p}_t(\mathbf{x}\mid \mathbf{z})\, q(\mathbf{z})\, d\mathbf{z} 
\right\rangle \, d\mathbf{x} \\
     &=  \frac{1}{M_t} \int_{\mathcal X} \, \left\langle \mathbf{u}_t^\theta (\mathbf{x}),\int_{\mathcal Z} \dot{\mathbf{\mu}}_t(\mathbf{z})\, m_t(\mathbf{z}) {p}_t(\mathbf{x}\mid \mathbf{z})\, q(\mathbf{z})\, d\mathbf{z} 
\right\rangle \, d\mathbf{x} \\
     &=  \frac{1}{M_t} \int_{\mathcal X} \int_{\mathcal Z}\, \left\langle \mathbf{u}_t^\theta(\mathbf{x}),\dot{\mathbf{\mu}}_t(\mathbf{z}) \right\rangle\, m_t(\mathbf{z}) {p}_t(\mathbf{x}\mid \mathbf{z})\, q(\mathbf{z})\, d\mathbf{z} 
 \, d\mathbf{x} \\
    & = \mathbb{E}_{q(\mathbf{z})}\big [ \frac{m_t(\mathbf{z})}{M_t}  \mathbb{E}_{p_t(\mathbf{x\vert \mathbf{z}})} \left\langle \mathbf{u}_t^\theta(\mathbf{x}),\dot{\bm{\mu}}_t(\mathbf{z}) \right\rangle \big ]
\end{align}

\noindent The remaining term $\|\mathbf{u}_t(\mathbf{x})\|^2$ does not depend on $\theta$ and therefore contributes only a constant. Combining the terms yields the
conditional objective $\mathcal L_{\mathrm{UCFM}}(\theta)$ up to a constant, which completes the proof. \hfill 
\end{proof}

\newpage

\subsection{Convexity analysis of the flow matching objective}

\begin{proposition}[Convexity of FM/CFM objective for a OU process with isotropic diffusion]
Consider the probability-flow velocity of the Ornstein-Ulhenbeck process with no growth dynamics
\[
\mathbf{u}_t^{A}(\mathbf{x})
=
-\mathbf{A}\mathbf{x}
-
D\mathbf{s}_t(\mathbf{x}),
\]
where \( \mathbf{A}\in\mathbb{R}^{d\times d} \), \(D>0\) is fixed, and
\(\mathbf{s}_t(\mathbf{x})=\nabla \log p_t(\mathbf{x})\) is fixed. Then the flow matching objective is convex in \(\mathbf{A}\). If
\[
\boldsymbol{\Sigma}_{\mathrm{FM}}
=
\mathbb{E}_{t\sim \mathcal{U}[0,T]}\mathbb{E}_{\mathbf{x}\sim p_t(\mathbf{x})}
\left[\mathbf{x}\mathbf{x}^{\top}\right]
\succ 0,
\]
then the objective is strictly convex and admits a unique minimizer. Moreover, since the FM and CFM objectives differ only by a constant independent of \(\mathbf{A}\), the corresponding CFM objective has the same convexity and minimizers with respect to \(\mathbf{A}\).
\end{proposition}
\begin{proof}
The FM objective \cite{lipman2022flow} is given by
\[
\mathcal{L}_{\mathrm{FM}}(\mathbf{A})
=
\mathbb{E}_{t \sim \mathcal{U}[0,T]}\mathbb{E}_{\mathbf{x}\sim p_t(\mathbf{x})}
\left[
\left\|
\mathbf{u}_t(\mathbf{x})
-
\mathbf{u}_t^{A}(\mathbf{x})
\right\|^2
\right],
\]
where, for the OU process, the probability flow velocity is
\[
\mathbf{u}_t^{A}(\mathbf{x})
=
-\mathbf{A}\mathbf{x}
-
D\mathbf{s}_t(\mathbf{x}).
\]
Substituting this expression into the FM objective yields
\[
\mathcal{L}_{\mathrm{FM}}(\mathbf{A})
=
\mathbb{E}_{t,\mathbf{x}}
\left[
\left\|
\mathbf{u}_t(\mathbf{x})
+
\mathbf{A}\mathbf{x}
+
D\mathbf{s}_t(\mathbf{x})
\right\|^2
\right] = \mathbb{E}_{t,\mathbf{x}}
\left[
\left\|
\mathbf{A}\mathbf{x}
+
\mathbf{r}_t(\mathbf{x})
\right\|^2
\right].
\]
where $\mathbf{r}_t(\mathbf{x})= \mathbf{u}_t(\mathbf{x})+ D\mathbf{s}_t(\mathbf{x})$. Note: For simply notation we refer to the expectation $\mathbb{E}_{t \sim \mathcal{U}[0,T]}\mathbb{E}_{\mathbf{x}\sim p_t(\mathbf{x})}$ with $\mathbb{E}_{t,\mathbf{x}}$.\\

\noindent Expanding the square term inside the expectation gives
\[
\left\|
\mathbf{A}\mathbf{x}
+
\mathbf{r}_t(\mathbf{x})
\right\|^2
=
\mathbf{x}^{\top}\mathbf{A}^{\top}\mathbf{A}\mathbf{x}
+
2\mathbf{r}_t(\mathbf{x})^{\top}\mathbf{A}\mathbf{x}
+
\|\mathbf{r}_t(\mathbf{x})\|^2.
\]
Taking expectations,
\[
\mathcal{L}_{\mathrm{FM}}(\mathbf{A})
=
\mathbb{E}_{t,\mathbf{x}}
\left[
\mathbf{x}^{\top}\mathbf{A}^{\top}\mathbf{A}\mathbf{x}
\right]
+
2\mathbb{E}_{t,\mathbf{x}}
\left[
\mathbf{r}_t(\mathbf{x})^{\top}\mathbf{A}\mathbf{x}
\right]
+
\mathbb{E}_{t,\mathbf{x}}
\left[
\|\mathbf{r}_t(\mathbf{x})\|^2
\right].
\]

\noindent Using the cyclic property of the trace, the quadratic term can be expressed as
\[
\mathbb{E}_{t,\mathbf{x}}
\left[
\mathbf{x}^{\top}\mathbf{A}^{\top}\mathbf{A}\mathbf{x}
\right]
=
\mathbb{E}_{t,\mathbf{x}}
\left[
\operatorname{Tr}
\left(
\mathbf{A}^{\top}\mathbf{A}\mathbf{x}\mathbf{x}^{\top}
\right)
\right] = 
\operatorname{Tr}
\left(
\mathbf{A}^{\top}\mathbf{A}
\boldsymbol{\Sigma}_{\mathrm{FM}}
\right),
\]
where
\[\boldsymbol{\Sigma}_{\mathrm{FM}}
= \mathbb{E}
_{t,\mathbf{x}} \left[ \mathbf{x} \mathbf{x}^\top \right].
\]
\noindent Using the cyclicity of the trace gives
\[
\operatorname{Tr}
\left(
\mathbf{A}^{\top}\mathbf{A}
\boldsymbol{\Sigma}_{\mathrm{FM}}
\right)
=
\operatorname{Tr}
\left(
\mathbf{A}\boldsymbol{\Sigma}_{\mathrm{FM}}\mathbf{A}^{\top}
\right).
\]
Similarly, the second term in the FM objective can we written as
\[
\mathbb{E}_{t,\mathbf{x}}
\left[
\mathbf{r}_t(\mathbf{x})^{\top}\mathbf{A}\mathbf{x}
\right]
=
\mathbb{E}_{t,\mathbf{x}}
\left[
\operatorname{Tr}
\left(
\mathbf{A}\mathbf{x}\mathbf{r}_t(\mathbf{x})^{\top}
\right)
\right] = \operatorname{Tr}(\mathbf{A} \mathbf{C}_\mathrm{FM})
\]
where $
\mathbf{C}_{\mathrm{FM}} = \mathbb{E}_{t,\mathbf{x}}
\left[
\mathbf{x}\mathbf{r}_t(\mathbf{x})^{\top}
\right] .
$ \\ 

\noindent We can rewrite the FM objective as a quadratic function of \(\mathbf{A}\),
\[
\mathcal{L}_{\mathrm{FM}}(\mathbf{A})
=
\operatorname{Tr}
\left(
\mathbf{A}\boldsymbol{\Sigma}_{\mathrm{FM}}\mathbf{A}^{\top}
\right)
+
2\operatorname{Tr}
\left(
\mathbf{A}\mathbf{C}_{\mathrm{FM}}
\right)
+
\mathrm{const},
\]
where the final term is independent of $\mathbf{A}$ and only depends on $\mathbf{r}_t(\mathbf{x})$.\\

\noindent We now prove convexity directly. Let
\(\mathbf{A}_1,\mathbf{A}_2\in\mathbb{R}^{d\times d}\) and
\(\lambda\in[0,1]\), and define $
\mathbf{A}_{\lambda} = \lambda \mathbf{A}_1 + (1-\lambda)\mathbf{A}_2.$
The linear term and the constant term cancel exactly in the convexity difference, so
\begin{align}
    \Delta & = 
\lambda \mathcal{L}_{\mathrm{FM}}(\mathbf{A}_1)
+
(1-\lambda)\mathcal{L}_{\mathrm{FM}}(\mathbf{A}_2)
-
\mathcal{L}_{\mathrm{FM}}(\mathbf{A}_{\lambda}) \\
& = 
\lambda
\operatorname{Tr}
\left(
\mathbf{A}_1\boldsymbol{\Sigma}_{\mathrm{FM}}\mathbf{A}_1^{\top}
\right)
+
(1-\lambda)
\operatorname{Tr}
\left(
\mathbf{A}_2\boldsymbol{\Sigma}_{\mathrm{FM}}\mathbf{A}_2^{\top}
\right)
-
\operatorname{Tr}
\left(
\mathbf{A}_{\lambda}
\boldsymbol{\Sigma}_{\mathrm{FM}}
\mathbf{A}_{\lambda}^{\top}
\right).
\end{align} 
Expanding \(\mathbf{A}_{\lambda}\) gives
\[
    \Delta 
=
\lambda(1-\lambda)
\operatorname{Tr}
\left(
(\mathbf{A}_1-\mathbf{A}_2)
\boldsymbol{\Sigma}_{\mathrm{FM}}
(\mathbf{A}_1-\mathbf{A}_2)^{\top}
\right) 
 =
\lambda(1-\lambda)\,
\operatorname{Tr}
\left(
\mathbf{H}\boldsymbol{\Sigma}_{\mathrm{FM}}\mathbf{H}^\top
\right),
\]
where $\mathbf{H} = \mathbf{A}_1 - \mathbf{A}_2$. \\

\noindent For $\boldsymbol{\Sigma}_{\mathrm{FM}}\succ 0$, and $\lambda (1-\lambda) \geq 0 $ the convexity gap 
\begin{equation}
\Delta = 
\lambda(1-\lambda)\operatorname{Tr}
\left(
\mathbf{H}\boldsymbol{\Sigma}_{\mathrm{FM}}\mathbf{H}^\top
\right)
=
\lambda(1-\lambda)\sum_{i=1}^d \mathbf{h}_i^\top \boldsymbol{\Sigma}_{\mathrm{FM}}\mathbf{h}_i \geq 0,
\end{equation}
is nonnegative, with \(\mathbf{h}_i\) being the rows of \(\mathbf{H}\). Moreover, if \(\boldsymbol{\Sigma}_{\mathrm{FM}}\succ 0\), then
\[
\operatorname{Tr}
\left(
\mathbf{H}\boldsymbol{\Sigma}_{\mathrm{FM}}\mathbf{H}^\top
\right)
> 0
\quad \text{for all } \mathbf{H}\neq 0,
\]
so \(\Delta = 0\) for \(\lambda\in(0,1)\) only when \(\mathbf{A}_1=\mathbf{A}_2\). Hence the objective is strictly convex and admits a unique minimizer.\\

\noindent Finally, the CFM objective can be expanded as 
\begin{align}
    \mathcal{L}_{\mathrm{CFM}}(\mathbf{A})
& =
\mathbb{E}_{t,\mathbf{z}\sim q(\mathbf{z})}
\mathbb{E}_{\mathbf{x}\sim p_t(\mathbf{x}\mid \mathbf{z})}
\left[
\left\|
\mathbf{u}_t(\mathbf{x}\mid \mathbf{z})
-
\mathbf{u}_t^{A}(\mathbf{x})
\right\|^2
\right],\\
& = \mathcal{L}_{\mathrm{FM}}(\mathbf{A})
+
\mathbb{E}_{t,z,\mathbf{x}}
\left[
\|\mathbf{u}_t(\mathbf{x}\mid \mathbf{z})\|^2
\right]
-
\mathbb{E}_{t,\mathbf{x}}
\left[
\|\mathbf{u}_t(\mathbf{x})\|^2
\right],
\end{align}
where the last two terms are independent of \(\mathbf{A}\). Therefore,
\[
\nabla_{\mathbf{A}}\mathcal{L}_{\mathrm{CFM}}(\mathbf{A})
=
\nabla_{\mathbf{A}}\mathcal{L}_{\mathrm{FM}}(\mathbf{A}),
\]
and CFM inherits the same convexity and minimizers with respect to \(\mathbf{A}\).
\end{proof}

\begin{corollary}[OU--Gaussian case]
For the OU process
\[
d\mathbf{x}_t=-\mathbf A\mathbf{x}_t\,dt+\sqrt{2D}\,d\mathbf W_t,
\]
with \(\mathbf A\succeq 0\) and fixed \(D>0\) and Gaussian initial distribution, the marginals remain Gaussian,
\[
p_t(\mathbf{x})=\mathcal{N}(\mathbf m_t,\mathbf{\Sigma}_t)
\quad \forall \: t.
\]
Hence the score function at time $t$ is $\mathbf{s}_t(\mathbf{x})=-\mathbf{\Sigma}_t^{-1}(\mathbf{x}-\mathbf m_t)$, and the FM second-moment matrix satisfies
\[
\mathbb{E}_{t}\mathbb{E}_{\mathbf{x}\sim p_t}
[\mathbf{x}\mathbf{x}^\top]
=
\mathbb{E}_t[\mathbf{\Sigma}_t+\mathbf m_t\mathbf m_t^\top].
\]
Consequently, using the linear parameterization of the probability-flow velocity, the FM and CFM objectives are convex quadratic functions of \(\mathbf A\), and are strictly convex whenever
\[
\mathbb{E}_t[\Sigma_t+\mathbf m_t\mathbf m_t^\top]\succ 0.
\]
The positive-semidefinite constraint \(\mathbf A\succeq 0\) ensures stability of the OU dynamics but is not required for convexity of the regression objective.
\end{corollary}

\newpage

\section{Simulation studies}

\subsection{Two-dimensional Ornstein--Uhlenbeck process}

We consider a two-dimensional Ornstein--Uhlenbeck (OU) process governed by the linear stochastic differential equation
\begin{equation}
    d\mathbf{x}_t = -\mathbf{B} \mathbf{x}_t\, dt + \sqrt{2D}\, d\mathbf{W}_t,
\end{equation}
where \( \mathbf{x}_t \in \mathbb{R}^2 \), \( \mathbf{B} \in \mathbb{R}^{2 \times 2} \) is a positive semi-definite drift matrix, and \( D > 0 \) is an isotropic diffusion coefficient. The drift matrix is constructed as \( \mathbf{B} = \mathbf{A A}^\top \), where \( \mathbf{A} \in \mathbb{R}^{2 \times 2} \) has entries sampled independently from a standard normal distribution. This ensures that \( B \) is symmetric positive semi-definite, yielding stable linear dynamics.\\

\noindent Trajectories are simulated using an Euler--Maruyama discretization,
\begin{equation}
    \mathbf{x}_{t+\Delta t} = \mathbf{x}_t - \Delta t\, \mathbf{B} \mathbf{x}_t + \sqrt{2D \Delta t}\, \boldsymbol{\xi}_t,
\end{equation}
where \( \boldsymbol{\xi}_t \sim \mathcal{N}(0, \mathbf{I}_2) \). We set \( D = 5 \), use a time step \( \Delta t = 0.01 \). and simulate trajectories for a total of \( T = 1 \).   For the drift matrix \(\mathbf{B}\) shown in Fig.~\ref{fig:ou-interpolation}, with eigenvalues approximately \(\lambda_1= 3.08\) and \(\lambda_2 = 0.16\), the dynamics are strongly anisotropic.\\


\noindent With initial conditions drawn from a Gaussian distribution centered at
\[
\mathbf{x}_0 \sim \mathcal{N}(\bm{\mu}_0, \sigma_0^2 \mathbf{I}_2), \quad \bm{\mu}_0 = (80,80), \: \sigma_0 = 4
\]
the system exhibits nontrivial stochastic dynamics. In particular, we observe drift-dominated motion along one eigendirection and diffusion-influenced dynamics along the other, as illustrated in Fig.~\ref{fig:ou-interpolation}. Cross-sectional snapshots (\(K=4\)) are recorded at regular intervals with spacing \(0.2\), yielding time-resolved snapshots \( p_{t_k}(\mathbf{x}) \).

\subsection{High-dimensional bifurcating stochastic dynamics}

\noindent We simulate stochastic dynamics in a high-dimensional bifurcating landscape governed by the Itô SDE
\begin{equation}
    d\mathbf{x} = -\nabla \phi(\mathbf{x})\, dt  + \sqrt{2D}\, d\mathbf{W}_t,
\end{equation}
where \( D > 0 \) is an isotropic diffusion coefficient. The potential \( \phi(\mathbf{x}) \) is defined through a collection of \(M\) soft attractors \( \{ \mathbf{a}_m \}_{m=1}^M \subset \mathbb{R}^d \),
\begin{equation}
    \phi(\mathbf{x}) = \sum_{m=1}^{M} \frac{w_m(\mathbf{x})}{2\sigma^2} \|\mathbf{x} - \mathbf{a}_m\|^2, \quad \text{where} \quad     w_m(\mathbf{x}) = \frac{\exp\left(-\|\mathbf{x} - \mathbf{a}_m\|^2 / (2\sigma^2)\right)}{\sum_{j=1}^M \exp\left(-\|\mathbf{x} - \mathbf{a}_j\|^2 / (2\sigma^2)\right)}.
\end{equation}
This yields a smooth multi-well landscape in which the induced force field corresponds to a weighted superposition of linear attractions toward the centers \( \mathbf{a}_k \). To generate synthetic cross-sectional data, we simulate the dynamics using an Euler-Maruyama discretization,
\begin{equation}
    \mathbf{x}_{t+\Delta t} = \mathbf{x}_t + \Delta t \left( -\nabla \phi(\mathbf{x}_t) \right) + \sqrt{2D \Delta t}\, \boldsymbol{\xi}_t,
\end{equation}
where \( \boldsymbol{\xi}_t \sim \mathcal{N}(0, \mathbf{I}_d) \). We set $D=0.5$, use a time step \( \Delta t = 0.01 \), and simulate trajectories up to a final time \( T = 1 \). The choice of the diffusion coefficient \(D\) and the attractor width \(\sigma\) sets the relative strength of stochastic and deterministic effects. For the parameters used here (\(D=0.5\), \(\sigma=0.5\)), a simple scaling analysis shows that drift- and diffusion-induced displacements are comparable over the observation time scale, placing the system in a regime where both contribute meaningfully to the dynamics. The attractors \( \{\mathbf{a}_m\}_{m=1}^M \) are chosen to be axis-aligned in \( \mathbb{R}^d \), inducing a combinatorial branching structure in which different subsets of coordinates define distinct lineage directions. Simulations are initialized with samples drawn from a narrow isotropic Gaussian centered at the origin,
$
\mathbf{x}_0 \sim \mathcal{N}(0, \sigma_0^2 \mathbf{I}_d)$ with $\sigma_0 = 0.01$.\\

\noindent Cross-sectional snapshots are recorded at \(K=5\) uniformly spaced time points with interval \(0.2\), yielding time-resolved samples from a high-dimensional stochastic branching process.

\section{Score estimation, validation, and flow matching details}
\label{sec:SM_score}
We employ denoising score matching (DSM) \cite{vincent2011connection, song2019generative} to estimate the score function of a high-dimensional empirical distribution. DSM proceeds by corrupting the input $\mathbf{x}$ with Gaussian noise: $\tilde{\mathbf{x}} \sim \mathcal{N}(\mathbf{x}, \sigma^2 \mathbf{I}_d)$. The resulting objective, aggregated over multiple noise scales, is given by:
\begin{equation}
\mathcal{\ell}(\phi, t_k) = 
\frac{1}{2} \sum_{i=1}^{L} \sigma_i^2 \, \mathbb{E}_{\mathbf{x}(t_k)} \, \mathbb{E}_{\tilde{\mathbf{x}} \sim \mathcal{N}(\mathbf{x}(t_k), \sigma_i^2 I)} 
\left[ \left\| \mathbf{s}_\phi(\tilde{\mathbf{x}}, \sigma_i, t_k) + \frac{\tilde{\mathbf{x}} - \mathbf{x}(t_k)}{\sigma_i^2} \right\|_2^2 \right].\label{eq:DSM}
\end{equation}
Here, the score function \( \mathbf{s}_\phi(\mathbf{x}, \sigma, t): \mathbb{R}^{d+1} \times [t_0,t_{K-1}] \to \mathbb{R}^d \) is parameterized by a feed-forward neural network with five hidden layers, each consisting of 100 nodes with ELU activations. For time-resolved cross-sectional data, we define the overall training objective by summing the losses across all time points:
\begin{equation}
\min_\phi \sum_{k=0}^{K-1} \lambda_k \, \mathcal{\ell}(\phi, t_k),
\end{equation}
where \( \lambda_k \) are adaptive weights \cite{maddu2022inverse} and \( t \sim \mathcal{U}(t_0, t_{K-1}) \). In all experiments, the noise levels \( \{ \sigma_i \}_{i=1}^L \) follow a geometric progression, with \( L = 5 \), \( \sigma_1 = 10 \), and \( \sigma_L = 0.01 \) \cite{song2019generative}. Optimization
is performed using the Adam optimizer with an initial learning rate $\eta = 10^{-3}$. We encourage readers to explore the accompanying code.\\

\paragraph*{Score validation:} After training a neural network to approximate the score function  \( \mathbf{s}_\theta(\mathbf{x},\sigma_L) \approx \nabla_{\mathbf{x}} \log p_{\text{data}}(\mathbf{x}) \), Langevin dynamics can be used to sample from the target distribution \( p_{\text{data}}(\mathbf{x}) \). Starting with a fixed step size \( \epsilon > 0 \) and an initial value \( \tilde{\mathbf{x}}_0 \sim \pi(\mathbf{x}) \), where \( \pi \) is a prior distribution, the Langevin method iteratively updates the samples using the equation:
\begin{equation}
    \tilde{\mathbf{x}}_\tau = \tilde{\mathbf{x}}_{\tau-1} + \frac{\epsilon}{2} \nabla_{\mathbf{x}} \log p(\tilde{\mathbf{x}}_{\tau-1}) + \sqrt{\epsilon} \mathbf{z}_\tau,
\end{equation}
where \( \mathbf{z}_\tau \sim \mathcal{N}(0, \mathbf{I}_d) \). Under certain conditions, as \( \epsilon \to 0 \) and \( \tau \to \infty \), the distribution of \( \tilde{\mathbf{x}}_\tau \) converges to \( p(\mathbf{x}) \), resulting in exact samples from \( p(\mathbf{x}) \) \cite{welling2011bayesian}. For finite \( \epsilon \) and \( \tau \), a Metropolis-Hastings update is often used to correct the approximation, though this correction is typically negligible for small \( \epsilon \) and large \( \tau \) \cite{song2019generative}. This sampling approach relies only on the score function \( \nabla_{\mathbf{x}} \log p(\mathbf{x}) \). We validate the score model by running Langevin dynamics to generate samples and evaluate how accurately they match the cross-sectional data.

\subsection*{Training of different models}

All models are trained using the probability flow matching procedure described in Algorithm~1, differing only in the choice of drift and diffusion parameterizations specified in Table~\ref{tab:models}. The force or the regulatory model $\mathbf{h}_\theta \in \mathbb{R}^d$ is parameterized by a feed-forward neural network with four spectrally normalized hidden layers, each consisting of 100 nodes with ELU activations. The spectral normalization  promotes stable training and controlled Lipschitz behavior of the learned dynamics.  Optimization
is performed using the Adam optimizer with an initial learning rate $\eta = 10^{-3}$.

\setlength{\tabcolsep}{1.25em}
{\renewcommand{\arraystretch}{1.25}
\begin{table}[h!]
\centering
\caption{Generative models of cell dynamics}
\label{tab:models}
\begin{tabular}{lll}
\hline
\textbf{Method} &  $\mathbf{f}(\mathbf{x},t)$ &  $\mathbf{D}(\mathbf{x},t)$ \\
\hline
CLE & $\mathbf{h(x)} - \ell  \mathbf{x}$ & $\frac{1}{2}\text{diag}\left(h_1(\mathbf{x}) + \ell x_1,\cdots, h_d(\mathbf{x}) + \ell x_d\right)$ \\
Multiplicative & $\mathbf{h(x)} - \ell  \mathbf{x}$  & $\sqrt{\mathbf{x}}$\\
Additive & $\mathbf{h(x)} - \ell  \mathbf{x}$  & $\mathbf{I}$ \\
ODE  & $\mathbf{h(x)} - \ell  \mathbf{x}$ & $0$\\
TrajNet++/TIGON++  &  $\mathbf{h}(\mathbf{x},t )- \ell  \mathbf{x}$  & $0$ \\
PRESCIENT   & $-\nabla \phi(\mathbf{x})- \ell  \mathbf{x}$  & $\kappa^2$ \\
\hline
\end{tabular}
\vspace{0.1em}
\end{table}

\vspace{1em}

\begin{figure*}
  \centering
  \includegraphics[width = 55em]{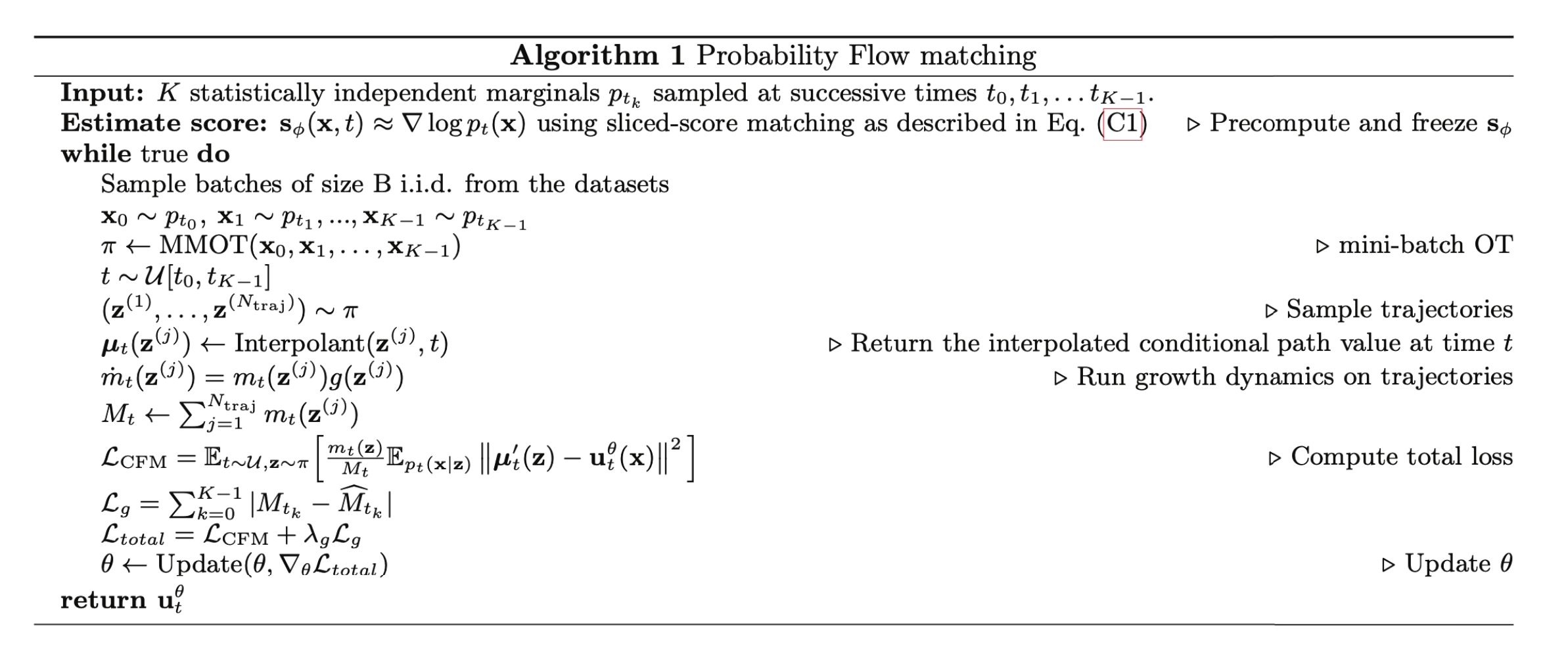}
\end{figure*}

\newpage
\begin{figure*}
  \centering
  \includegraphics[width = 52em]{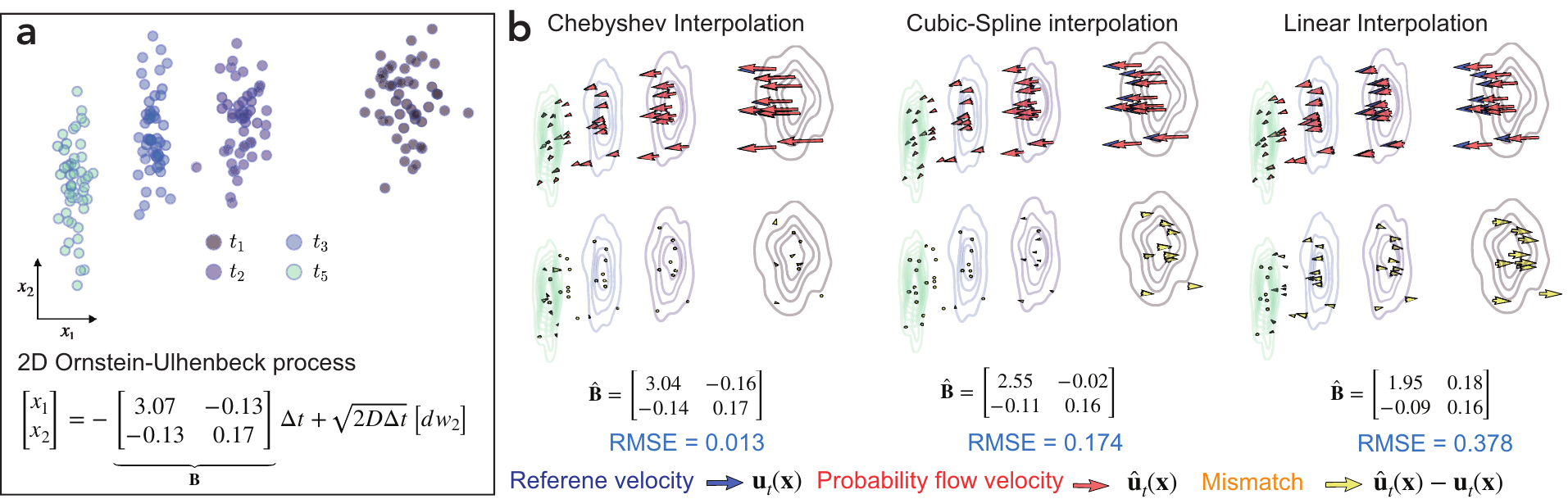}
\caption{
    \textbf{Evaluating the impact of interpolation schemes on drift inference from time-resolved snapshot data of a 2D Ornstein-Uhlenbeck process.}
    \textbf{(a)} Snapshots of particle positions at five time points $(t_1, t_2, \ldots, t_4)$ drawn from a 2D Ornstein--Uhlenbeck process governed by the linear drift matrix $\mathbf{B}$ and stochastic noise.
    \textbf{(b)} Comparison of Chebyshev, cubic spline, and linear interpolation methods to evaluate the conditional flow velocities to regress against in conditional flow matching. Top row: Probability flow velocities $\hat{\mathbf{u}}_t(\mathbf{x})$ (red arrows) overlaid on probability density contours at each time point. Bottom row: Inferred drift matrices $\hat{\mathbf{B}}$ and corresponding root mean squared errors (RMSE). Chebyshev interpolation yields the most accurate estimates, with the lowest RMSE and minimal mismatch between conditional velocity $\mathbf{u}_t(\mathbf{x})$ (blue arrows) and estimated flow $\hat{\mathbf{u}}_t(\mathbf{x})$.
    Yellow arrows represent the mismatch $\hat{\mathbf{u}}_t(\mathbf{x}) - \mathbf{u}_t(\mathbf{x})$. }
    \label{fig:ou-interpolation}
\end{figure*}

\begin{figure*}
  \centering
  \includegraphics[width = 16cm]{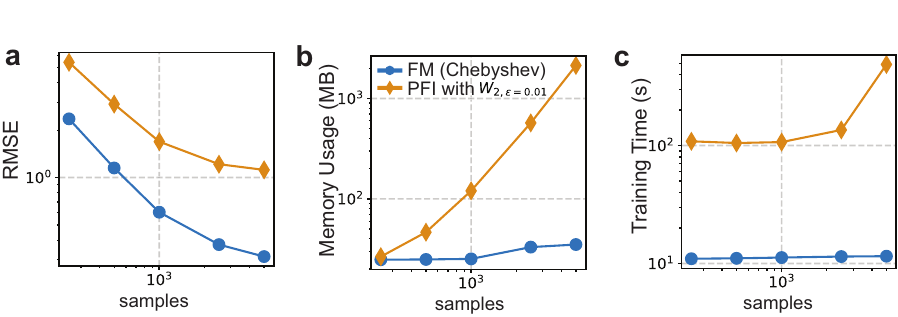}
\caption{
    \textbf{Simulation-based versus simulation-free inference.} Comparison of simulation-based Probability Flow Inference (PFI) and simulation-free Probability Flow Matching (PFM) on the differentiation landscape problem described in the main text (Fig.~\ref{fig:Bifur-interpolation}). In PFM, the mean path $\mu_t$ is parameterized using Chebyshev interpolants. (a) RMSE of the inferred force field ($\nabla \phi$) as a function of samples per snapshot; (b) corresponding memory usage; (c) training time.}
    \label{fig:compute}
\end{figure*}

\begin{figure*}
  \centering
  \includegraphics[width = \textwidth]{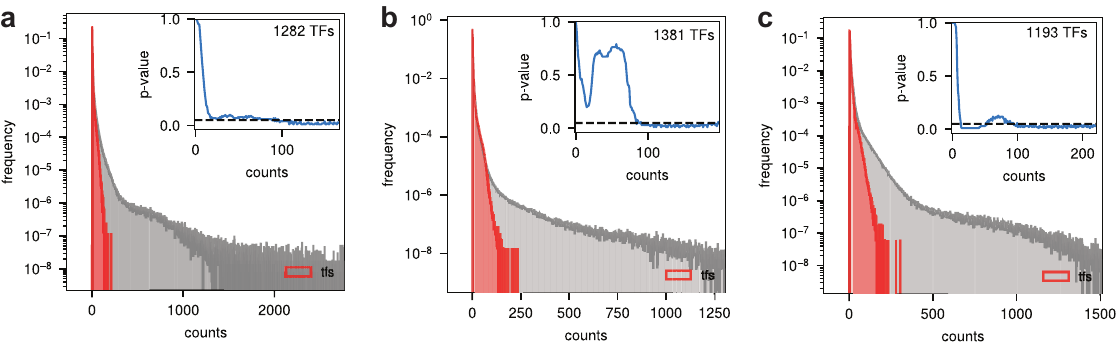}
\caption{\textbf{Distributions of counts for TF genes and non-TF genes.} In grey is the average count distribution over random sets of non-TF genes of the same size as the set of TF genes used compute the count distribution in red. The number of TF genes used is shown in the inset. The inset shows the p-value for each bin in the count distribution. For each bin, the p-value quantifies how often a random set of non-TF genes had less counts than the set of TF genes. The dashed line is the $0.05$ threshold.
 \textbf{(a)} Ex-vivo differentiation dataset \cite{georgolopoulos2021discrete} \textbf{(b)} Human CITE-seq \cite{demeo2025active} HSC differentiation dataset \textbf{(c)} Mouse HSC differentiation dataset \cite{weinreb2020lineage}.}
\label{fig:tfs_rest}
\end{figure*}

\begin{figure}
  \centering
  \includegraphics[width = 14cm]{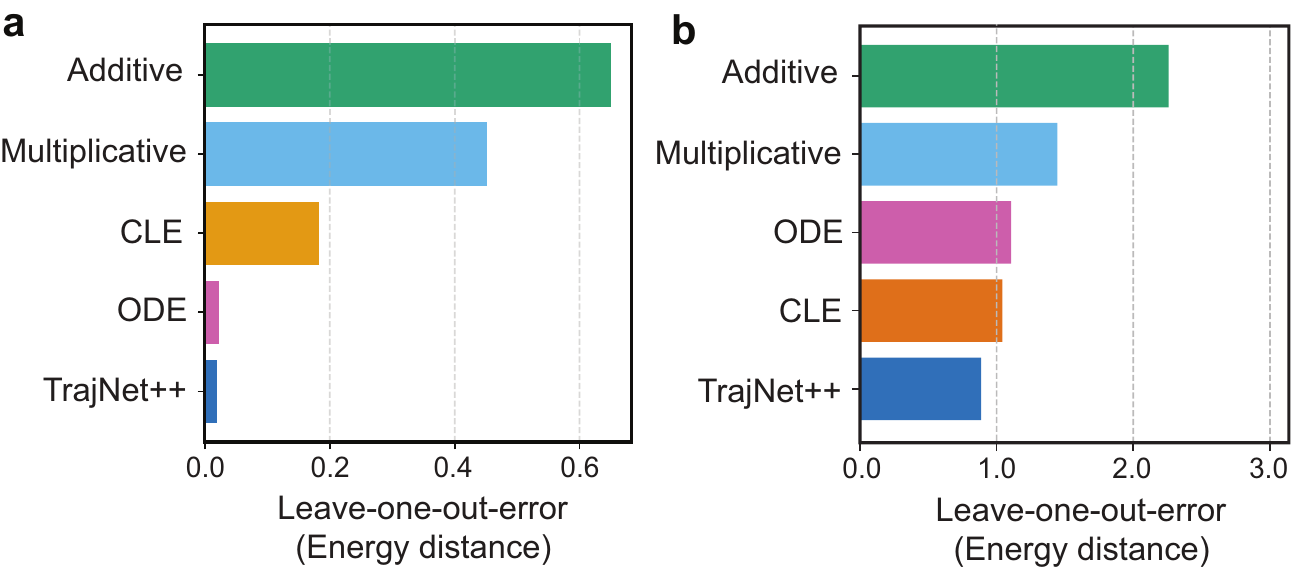}
\caption{
    \textbf{Interpolation accuracy:} (a) Interpolation accuracy, quantified by leave-one-out energy distance, for different dynamical models fitted to the 6-day \textit{in vitro} Mouse HSC differentiation dataset \cite{weinreb2020lineage}. (b) Interpolation accuracy on the 10-day \textit{in vitro} CITE-seq data capturing HSC differentiation\cite{demeo2025active}. 
    }
    \label{fig:interpolation_error}
\end{figure}

\begin{figure*}
  \centering
  \includegraphics[width = 35em]{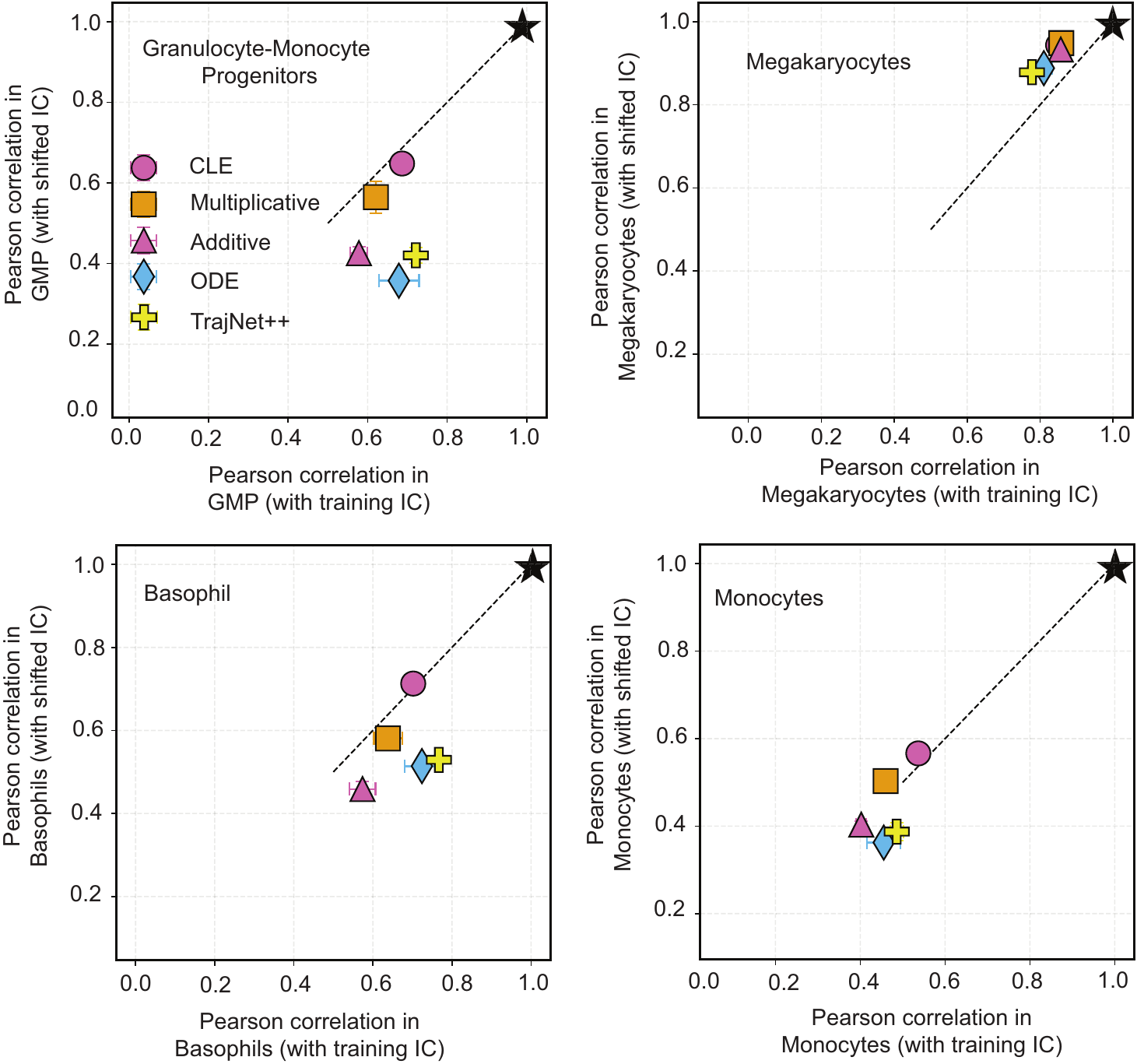}
\caption{\textbf{Generalization results.} Comparison of generalization performance across models trained on human CITE-seq dataset \cite{demeo2025active} and tested for generalization on the \textit{ex vivo} data \cite{georgolopoulos2021discrete}. The plots show the Pearson correlation between training and test performance when initialized from HSC/HSPC states across different cell-types present in the train data. The star symbol indicates perfect generalization.
}
\label{fig:kaggle_generalization}
\end{figure*}

\begin{figure}
  \centering
  \includegraphics[width = 38em]{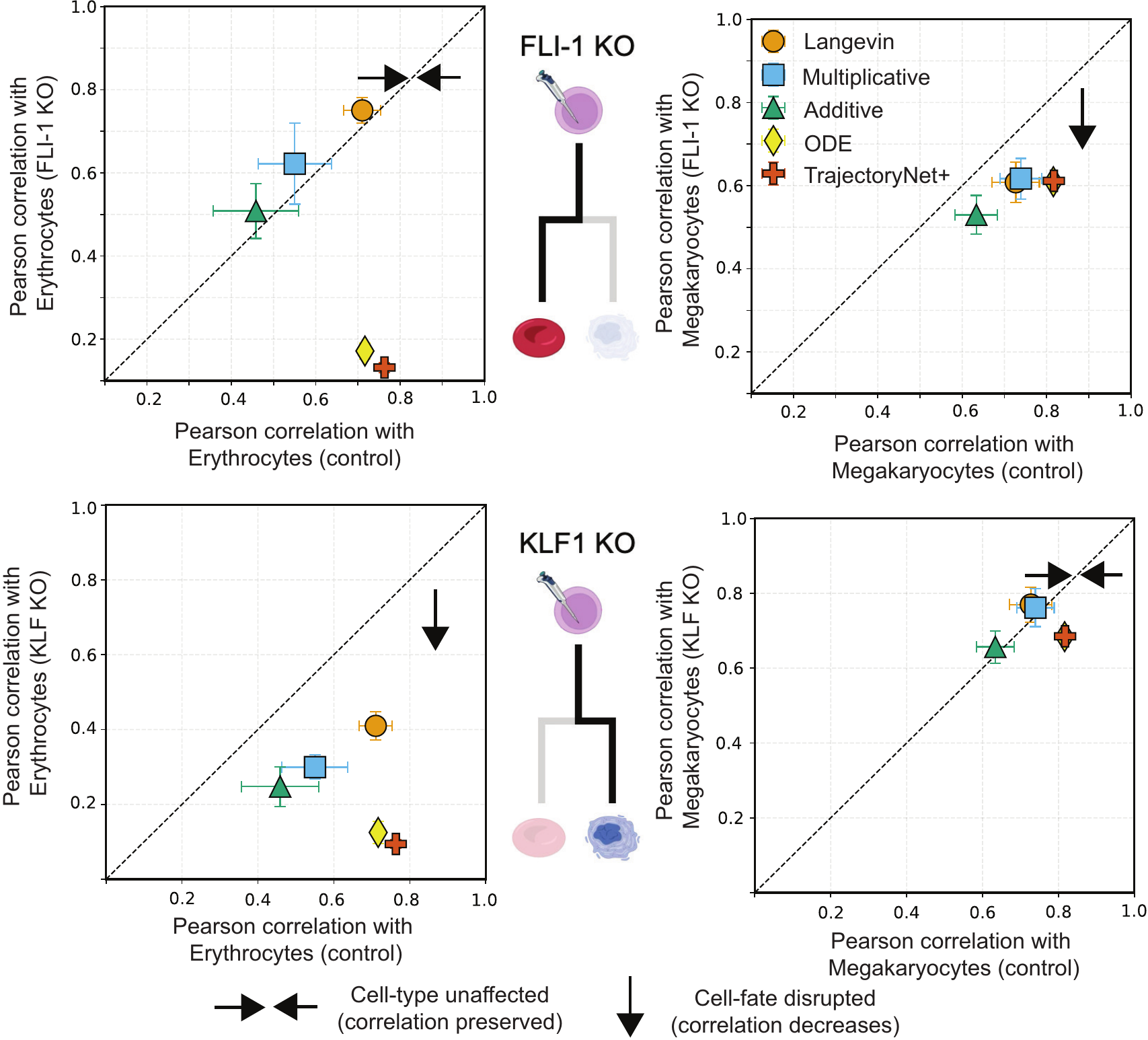}
\caption{
    \textbf{In-silico perturbation of the inferred transcriptome model:}
    \textbf{(a)} Each panel compares the Pearson correlation of knockout (KO) populations with control Erythrocytes and Megakaryocytes. KLF1-KO denotes knockout of the KLF1 gene, which impairs erythroid differentiation but leaves the megakaryocytic lineage largely intact. FLI1-KO denotes knockout of the FLI1 gene, which disrupts megakaryocyte formation while preserving erythroid identity. Arrows indicate these expected (``ground-truth'') biological outcomes. Points near the diagonal correspond to preserved correlations (unaffected fates), whereas points below the diagonal indicate reduced correlations and lineage disruption.
    }
    \label{fig:KO}
\end{figure}

\begin{table}[h!]
\centering
\caption{Abbreviations of cell types used in Fig.~3C.}
\label{tab:cell_types}
\begin{tabular}{ll}
\hline
\textbf{Abbreviation} & \textbf{Cell type} \\
\hline
HSC  & Hematopoietic stem cell \\
HSPC & Hematopoietic stem and progenitor cell \\
MEMP & Megakaryocyte-erythroid-myeloid progenitor \\
MPP  & Multipotent progenitor \\
Ery  & Erythrocytes \\
Mk   & Megakaryocytes \\
MPP1 & MPPs with megakaryocytic signature \\
MPP2 & MPPs with erythroid signature \\
Mk1  & Cells with early megakaryocytic signature \\
Mk2  & Mature megakaryocytes \\
Ery1 & Cells with early erythroid signature \\
Ery2 & Mature erythrocytes \\
\hline
\end{tabular}
\end{table}

\begin{figure}
  \centering
  \includegraphics[width = 57em]{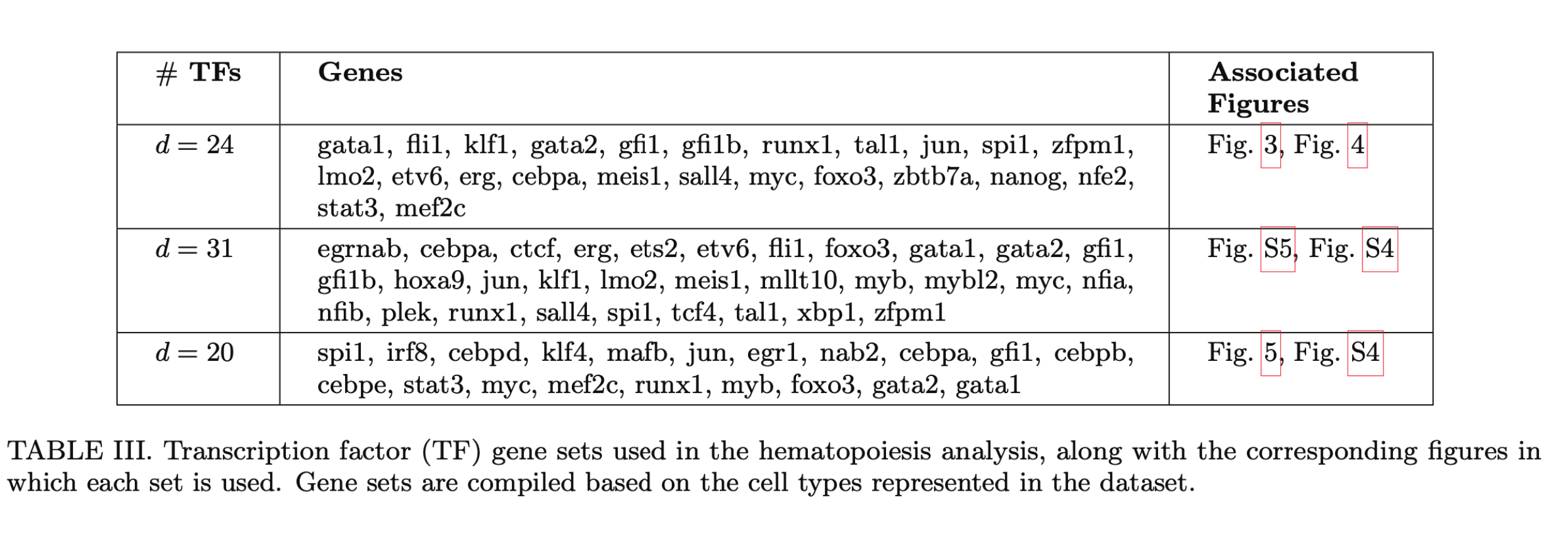}
\end{figure}


\end{document}